\documentclass[12pt]{article}
\usepackage{amsmath}
\usepackage{graphicx,psfrag,epsf}
\usepackage{enumerate}
\usepackage{natbib}
\usepackage{geometry}
\usepackage{bm}
\usepackage{epstopdf}
\usepackage{easybmat}
\usepackage{etex}
\usepackage{dsfont}
\usepackage{tabularx}
\usepackage{amssymb}
\usepackage{hyperref}
\usepackage{stmaryrd}
\usepackage{amsthm}
\usepackage{amsmath}%
\usepackage{amsfonts}%
\usepackage{amssymb}%
\usepackage{array}
\usepackage{booktabs}
\usepackage[usenames]{color}
\usepackage{epstopdf}
\usepackage{caption}
\usepackage[labelformat=simple]{subcaption}

\usepackage{float}
\usepackage{algorithm}
\usepackage{algpseudocode}
\usepackage{natbib}
\usepackage{nomencl}
\usepackage[font={small}]{caption}
\usepackage{enumitem}
\usepackage{ulem}
\usepackage{bbm}
\usepackage[nodisplayskipstretch]{setspace}
\usepackage{hyperref}
\newtheorem{theorem}{Theorem}[section]

\newtheorem{proposition}[theorem]{Proposition}

\numberwithin{equation}{section}
\algrenewcommand\alglinenumber[1]{\tiny #1:}
\newcommand\restr[2]{{
  \left.\kern-\nulldelimiterspace 
  #1 
  \vphantom{\big|} 
  \right|_{#2} 
  }}
\graphicspath{./Figures/}

\newcommand{\argmin}[1]{\underset{#1}{\operatorname{argmin}}\;}

\newcommand{\ES}[1]{\textnormal{\tiny {#1}}}
\newcommand{\vvec}{\operatorname{vec}}

\newcommand{\dual}{Z}
\renewcommand{\arraystretch}{1.2}
\newgeometry{tmargin=3cm, bmargin=3cm}
\definecolor{damian}{RGB}{0,0,250}
\definecolor{xixi}{RGB}{50,200,50}
\definecolor{jarek}{RGB}{255,0,0}
\definecolor{tim}{RGB}{160,0,255}

\newcommand{\T}[1]{^\mathsf{#1}}

\newcommand{\SP}{\textrm{SpINNEr}}

\makenomenclature

\floatname{algorithm}{Procedure}

\newcommand{\blind}{0}

\begin{document}

\def\spacingset#1{\renewcommand{\baselinestretch}%
{#1}\small\normalsize} \spacingset{1}

\if0\blind
{
  \title{\bf  A Sparsity Inducing Nuclear-Norm Estimator (SpINNEr) for Matrix-Variate Regression in Brain Connectivity Analysis}
  \author{Damian Brzyski$^a$, Xixi Hu$^b$, Joaquin Goni$^c$,	Beau Ances$^d$,\\
  Timothy W Randolph$^e$, Jaroslaw Harezlak$^f$\\
	\vspace{-3pt}
    $\mbox{}^{a}${\it \small Faculty of Pure and Applied Mathematics, Wroclaw University of Science and Technology, Wroclaw, Poland}\\
		\vspace{-3pt}
		$\mbox{}^{b}${\it \small Department of Statistics,  Indiana University, Bloomington,  IN, USA}\\
		\vspace{-3pt}
		$\mbox{}^{c}${\it \small Purdue University, West Lafayette, IN, USA}\\
		\vspace{-3pt}
		$\mbox{}^{d}${\it \small Washington University School of Medicine, St. Louis, MO, USA}\\
		\vspace{-3pt}
		$\mbox{}^{e}${\it \small Fred Hutchinson Cancer Research Center, Seattle, WA, USA}\\
		\vspace{-3pt}
		$\mbox{}^{f}${\it \small Department of Epidemiology and Biostatistics,  Indiana University, Bloomington, IN, USA}\\
   }
  \maketitle
} \fi

\if1\blind
{
  \bigskip
  \bigskip
  \bigskip
  \begin{center}
    {\LARGE\bf Nuclear and L1 Norms Marriage in Sparse and Low-Rank Regularized Matrix Estimation}
\end{center}
  \medskip
} \fi

\vspace{7pt}
\begin{abstract}
Classical scalar-response regression methods treat covariates as a vector and estimate a corresponding vector of regression coefficients. In medical applications, however, regressors are often in a form of multi-dimensional arrays. For example, one may be interested in using MRI imaging to identify which brain regions are associated with a health outcome. Vectorizing the two-dimensional image arrays is an unsatisfactory approach since it destroys the inherent spatial structure of the images and can be computationally challenging. We present an alternative approach---regularized matrix regression---where the matrix of regression coefficients is defined as a solution to the specific optimization problem. The method, called SParsity Inducing Nuclear Norm EstimatoR (SpINNEr), simultaneously imposes two penalty types on the regression coefficient matrix---the nuclear norm and the lasso norm---to encourage a low rank matrix solution that also has entry-wise sparsity. A specific implementation of the alternating direction method of multipliers (ADMM) is used to build a fast and efficient numerical solver. Our simulations show that SpINNEr outperforms others methods in estimation accuracy when the response-related entries (representing the brain's functional connectivity) are arranged in well-connected communities. SpINNEr is applied to investigate associations between HIV-related outcomes and functional connectivity in the human brain.
\end{abstract}

\noindent%

\spacingset{1.45}

\section{Introduction}

Regression problems where the response is a scalar and the predictors constitute a multidimensional array arise often in medical applications where a matrix or a high dimensional array of measurements is collected for each subject. For example, it is of clinical interest to understand associations between: (a) alcoholism and the electrical activity of different brain regions over time collected from electroencephalography (EEG) \citep{Li-dimension-2010}; (b) cognitive function and three-dimensional white-matter structure data collected from diffusion tensor imaging (DTI) \citep{Goldsmith-smooth-2014} for patients with multiple sclerosis (MS); and (c) cognitive impairment and brain's metabolic activity data collected from three-dimensional positron emission tomography (PET) imaging \citep{Wang-regularized-2014}.  Our work focuses on the problem of identifying brain network connections that are associated with neurocognitive measures for HIV-infected individuals. The outcome (response) is a continuous variable and the predictors are matrix representations of functional connectivity between the brain's cortical regions.

Biophysical considerations motivate our interest in estimating a matrix of regression coefficients that has the following two properties: (i) it should be relatively sparse, since we aim to identify connections that most strongly predict the outcome; and more importantly, (ii) the response-related connections form clusters, since brain activity networks are known to consist of densely connected regions. These two properties translate to the coefficient matrix having relatively small clusters, or blocks of nonzero entries, which implies that it is low-rank. Hence, we aim to solve the matrix regression problem by estimating a coefficient matrix that is both sparse and low-rank. To further illustrate our approach, consider the three matrices in Figure~\ref{fig_LowRankSparseMatrices}. The one in the left panel is sparse, but full-rank, the one on the right panel is low-rank, but not sparse, while the one in the middle panel is both low-rank and sparse, which is the structure we are interested in. To find such a solution, we propose a regularization method called \textit{SParsity Inducing Nuclear Norm EstimatoR} ($\SP$).
\begin{figure}[ht]
\centering
\begin{subfigure}{.32\textwidth}
  \centering
  \includegraphics[width=1\linewidth]{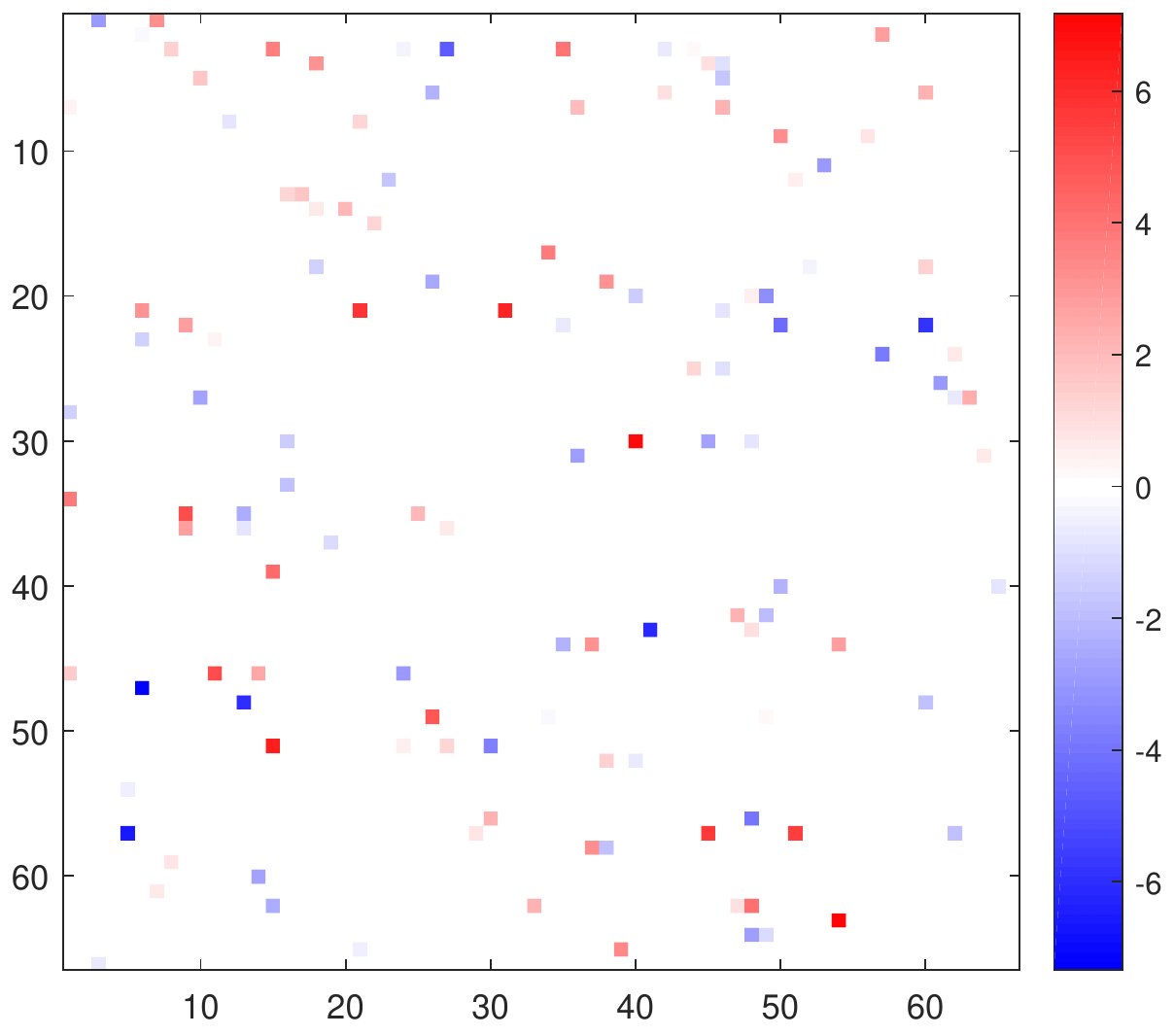}
  \caption{Sparse, full-rank}
\end{subfigure}%
\begin{subfigure}{.32\textwidth}
  \centering
  \includegraphics[width=1\linewidth]{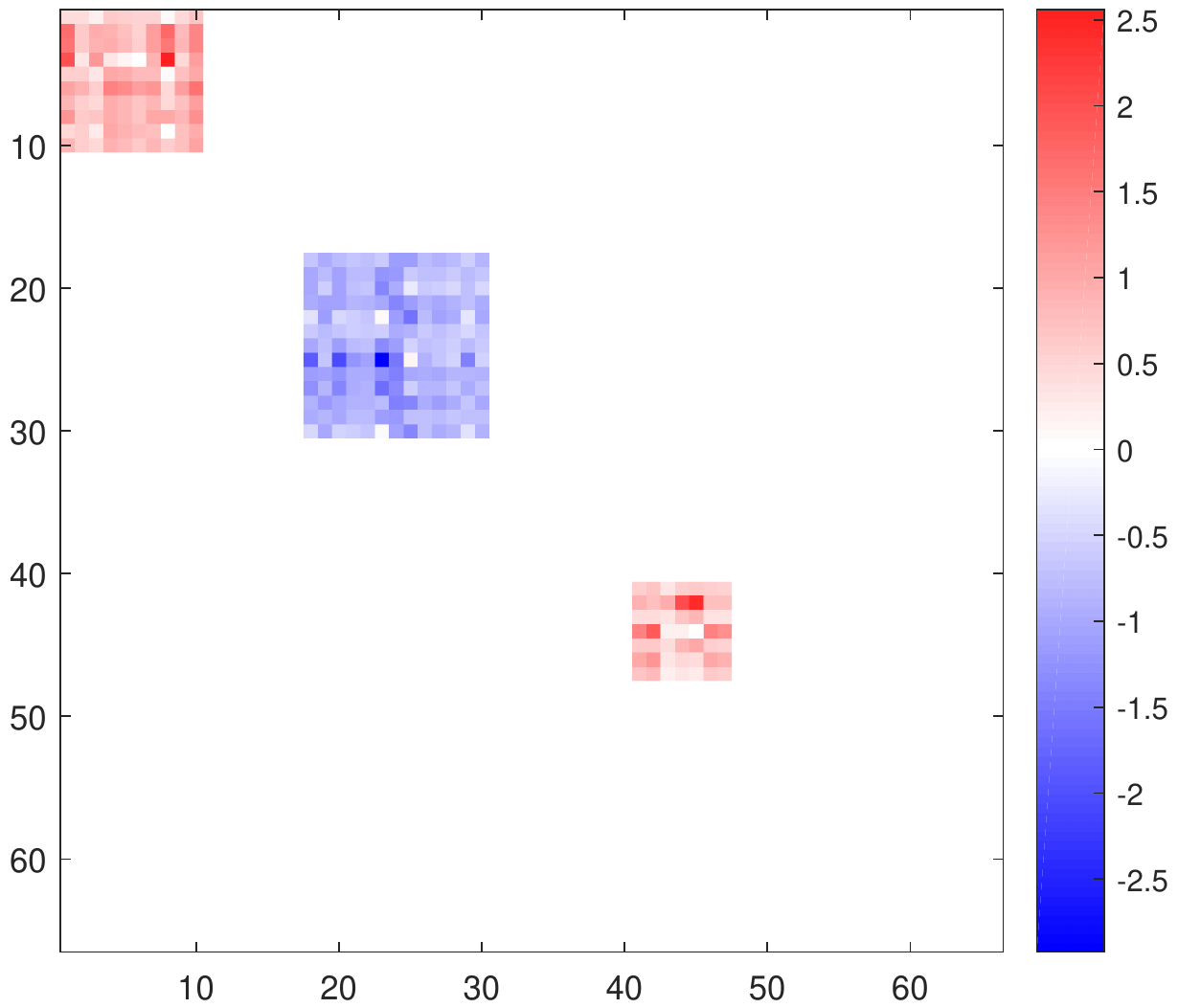}
  \caption{Low-rank and sparse}
\end{subfigure}
\begin{subfigure}{.32\textwidth}
  \centering
  \includegraphics[width=1\linewidth]{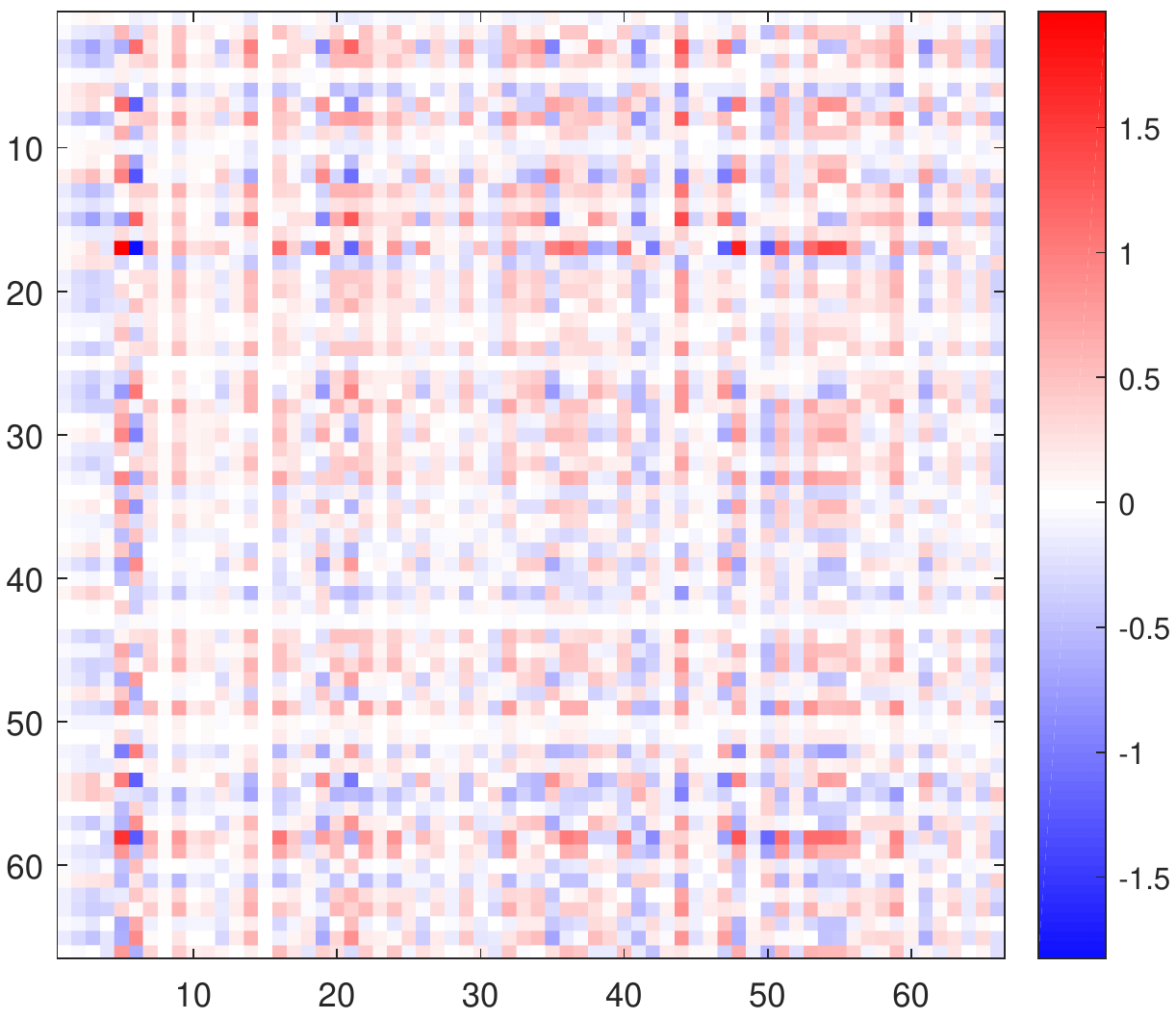}
  \caption{Low-rank, not sparse}
\end{subfigure}
\caption{Illustrative example of (a) a matrix that is sparse, but full-rank; (c) a matrix that is low-rank, but not sparse; and (b) a matrix that is both low-rank and sparse.}
\label{fig_LowRankSparseMatrices}
\end{figure}

Several regularization methods have been proposed for regression problems where the response is a scalar and the predictors constitute a multidimensional array or tensor. These methods fall mainly along two investigative directions. The first treats the multidimensional array of predictors as functional data. Among the early efforts, \cite{Reiss-functional-2010} extended their functional principal component regression method for one-dimensional signal predictors \citep{Reiss-functional-2007} to two-dimensional image predictors. Their method is based on B-splines with a penalty on the roughness of the coefficient function which encourages local structure but does not impose constraints on rank or sparsity. \cite{Wang-regularized-2014} developed a regularized wavelet-based approach that induces sparsity in the coefficient function. The second line of research treats images as tensors rather than as functional data. \cite{Zhou-tensor-2013} proposed a tensor regression framework that achieves dimension reduction through fixed-rank tensor decomposition \citep{Kolda-tensor-2009}. They obtain the estimates by a regularized maximum likelihood approach. For regression problems with matrix covariates, \cite{Zhou-regularized-2014} proposed spectral regularization where the penalty term is a function of the coefficient matrix singular values. Using $\ell_1$-norm of the singular values as penalty gives rise to a nuclear norm regression which induces a low-rank structure on the coefficient matrix. Our work builds upon \cite{Zhou-regularized-2014} by inducing sparsity of the coefficient matrix in terms of both its rank (low-rank) and the number of its nonzero entries (sparse). Under a Bayesian framework, \cite{Goldsmith-smooth-2014} used a prior distribution on latent binary indicators to induce sparsity and spatial contiguity of relevant image locations, and appealed to Gaussian Markov random field to induce smoothness in the coefficients.

The approaches summarized above are insufficient for finding a coefficient matrix that is both sparse and low-rank. More specifically, for regularization approaches based on functional regression, \cite{Reiss-functional-2010} do not impose conditions on sparsity rank, while \cite{Wang-regularized-2014} do not impose any constraint on rank. For approaches based on tensor decomposition, the method of \cite{Zhou-tensor-2013} can potentially induce sparsity via regularized maximum-likelihood, but the rank of the solution must be pre-specified and fixed prior to model fitting. In other words, the rank is not determined in a data-driven manner. \cite{Zhou-regularized-2014} lifted the fixed-rank constraint by using a nuclear norm---a convex relaxation of rank---as penalty, but the solution may not be sparse. Finally, \cite{Goldsmith-smooth-2014} impose sparsity and spatial smoothness, which implicitly reduces complexity (and possibly rank), but this approach assumes spatially adjacent regions are similarly associated with the response.

In contrast to all of these methods, $\SP$ combines a nuclear-norm penalty with an $\ell_1$-norm penalty that simultaneously imposes low-rank and sparsity on the coefficient matrix. Specifically, the low-rank constraint induces accurate estimation of coefficients inside response-related blocks, while outside of these blocks the sparsity constraint encourages zeros. These blocks, however, are not presumed to consist of only spatially adjacent brain regions and so this sparse-and-low-rank approach is more flexible and physiologically meaningful.

While $\SP$ seeks a singly regression coefficient matrix that is both sparse and low-rank, others have proposed estimating two structures: one low-rank and one sparse. For graphical models, in particular, \cite{Chandrasekaran-latent-2012} proposed a method for estimating a precision matrix that decomposes into the sum of a sparse matrix and a low-rank matrix when there are latent variables. Their estimation is based on regularized maximum likelihood where sparsity is induced by the $\ell_1$-norm and low-rank is induced by the nuclear norm. Building upon \cite{Chandrasekaran-latent-2012}, \cite{Ciccone-robust-2019} imposed the additional constraint that the sample covariance matrix of the observed variables is close to the true covariance in terms of Kullback-Leibler divergence, and proposed a computational solution based on the alternating direction method of multipliers (ADMM) algorithm. While \cite{Chandrasekaran-latent-2012} and \cite{Ciccone-robust-2019} assume observations to be i.i.d., \cite{Foti-sparse-2016} extended the ``sparse plus low-rank'' framework to graphical models with time series data, whereas \cite{Basu-low-rank-2018} considered vector autoregressive models and directly imposed the decomposition on the transition matrix. In the matrix completion literature, the ``sparse plus low-rank'' decomposition has also been exploited for algorithmic concerns, enabling more efficient storage and computation \cite{Mazumder-spectral-2010, Hastie-matrix-2015}. Our work is distinct from these proposals in that we obtain a single matrix that is simultaneously sparse and low-rank by imposing two penalties on the same matrix, rather than separately penalizing two components of a matrix.

The rest of the article is organized as follows. Section~\ref{sec_Problem} further motivates and describes the objective of finding response-related clusters and translate it to a problem of finding a low-rank and sparse coefficient matrix. Section~\ref{sec_Method} formulates the objective as an optimization problem, characterizes properties of its solution, and develops an algorithm for numerical implementation. Simulation experiments are summarized in Section~\ref{sec_SimulationExperiments} and an application to brain imaging data is described in Section~\ref{sec_BrainDataResults}. We conclude with discussion in Section~\ref{sec_Discussion}. Technical derivations of the algorithm are presented in the Appendix.

\section{Clusters recovery problem}
\label{sec_Problem}
\subsection{Statistical model}
Assume we observe a real-valued response, $y_i$, and a $p\times p$ matrix, $A_i$, for each subject, $i=1,\ldots,n$. We additionally assume a vector of $m$ covariates, $X_i$, such that the $n\times m$ matrix, $X$, with rows $X_i$, for each subject, has independent columns (hence $m\leq n$).

Motivated by brain imaging applications, $A_i$ is viewed as an adjacency matrix of connectivity information (structural or functional), each $X_i$ corresponds to a vector of demographic covariates and an intercept (i.e., the first entry of $X_i$ is 1). Additionally, we assume that there exists an (unknown) $p\times p$ matrix $B$ and (unknown) $m\times 1$ vector $\beta$ whose entries represent coefficients to be estimated in the regression equation
\begin{equation}
\label{model}
y_i = \langle A_i, B \rangle + X_i \beta   + \varepsilon_i, \quad \textrm{for } i=1,\ldots, n,
\end{equation}
where $\langle A_i, B \rangle:= \big\langle \vvec(A_i),\, \vvec(B) \big\rangle = \operatorname{tr}\big(A_i\T{T}B\big)$ is the Frobenius inner product and $\varepsilon\sim\mathcal{N}(0,\sigma^2I_n)$.
Unless $n$ is unusually large (greater than $p(p-1)/2$\,), this is an informal statement of the problem which does not have a unique solution for $B$ and $\beta$ without further constraints. The focus of this work is on the rigorous implementation of constraints that lead to a biologically meaningful regression model having a unique solution.

We will use the equivalent graph description of the problem to build the intuition behind the assumed model \eqref{model}. In that interpretation, brain regions are viewed as $p$ nodes in a graph and the connectivity information of $i$th subject for these regions is represented by a $p\times p$ matrix $A_i$ with zeros on the diagonal.  The off-diagonal entries, $A_i(j,l)$, of $A_i$ denote weights of connectivity between regions $j$ and $l$; positive as well as negative weights are acceptable. If $A_i(j,l)$ is positive, its value indicates how strongly regions $j$ and $l$ are connected, while the magnitude of negative entry indicates the level of dissimilarity.  In \eqref{model}, $B$ denotes the (unknown) $p\times p$ matrix of regression coefficients whose $(j,l)$ entry, $B_{j,l}$, represents the association between the response variable and the connectivity across regions $j$ and $l$. We assume $B$ is symmetric and note that its diagonal entries $B_{j,j}$ are not included in the model, since each connectivity matrix $A_i$ has zeros on the diagonal. The main goal, therefore, is to estimate the off-diagonal entries of $B$ in a manner that reveals brain subnetwork structure that is associated with the response.  This structure is revealed by the clusters and hubs defined by the non-zero entries in $\hat{B}$, an estimate of $B$.

\subsection{Response-related clusters}
We are interested in identifying only significant brain-region connectivities and therefore we want to encourage the estimate, $\hat{B}$, to be sparse entry-wise. However, our most important goal is to protect the structure of response-related connectivities (i.e., the non-zero entries of $B$). Indeed, brain networks exhibit a so-called ``rich club" organization, meaning that there are relatively small groups of densely connected nodes \citep{Richclub, HubDet}. Recent studies have demonstrated that such hubs play an important role in information integration between different parts of the network \citep{Richclub}. This structure would be lost if the process of estimating $B$ only focused on sparsity. Thus, we want the estimation process to allow for a potentially cluster-structured form of $B$ so that it may more accurately reflect the association between brain connectivities and a phenotypic outcome. However, estimating $B$ using a two-step process would imply that we must detect individual response-impacting edges \textit{prior to} forming clusters from the selected edges. Clearly some cluster-defining edges and/or hubs may be missed by this preliminary focus on sparsity.

As an illustration, consider a setting where there are many response-related connectivities between a few, say $k$, brain regions. Even if some of these effects are moderate, it is the entire cluster of regions that, as a whole, affects the response. However, if there are $(k-1)k/2$ connections in the cluster and only the strongest effects survive a sparsity-inducing lasso estimate \citep{LassoF} or other entry-wise thresholding techniques, then the ``systems level" information is lost and inferring relevant information about the clusters may become impossible.

In our work, we introduce the notion of \textit{response-related clusters} and focus on their selection rather than on accurate estimation of each individual effect which, in fact, would be impossible due to the limited sample size. Precisely, we define a set of nodes, $S$, to be a \textit{response-related cluster} (RRC) if for any two distinct indices $j, l \in S$ there is a path of edges from $S$ connecting $j$ and $l$, namely the sequence of the elements $i_1, \ldots, i_k \in S$ such as $i_1 = j$, $i_k= l$ and $B_{i_h, i_{h+1}}\neq 0$ for $h=1,\ldots, k-1$. We define $S$ to be a \textit{positive response-related cluster} if for every pair of its elements, $j$ and $l$, it holds that $B_{j,l}\geq 0$ (zeros are acceptable). Accordingly, $S$ is a \textit{negative response-related cluster} if $B_{j,l}\leq 0$ for all its elements. Motivated by the rich-club pattern of brain connectivity, we will assume that relatively few clusters of brain nodes spanning the subnetworks are strongly associated with $y$. If the brain regions are arranged in a cluster-by-cluster ordering, this assumption is simply reflected in a block-diagonal pattern of the matrix of regression coefficients, $B$, with blocks corresponding to RRCs (Figure \ref{fig:signalForm}\,).
\begin{figure}[ht]
\centering
\begin{subfigure}{.34\textwidth}
  \centering
	\includegraphics[width=1\linewidth]{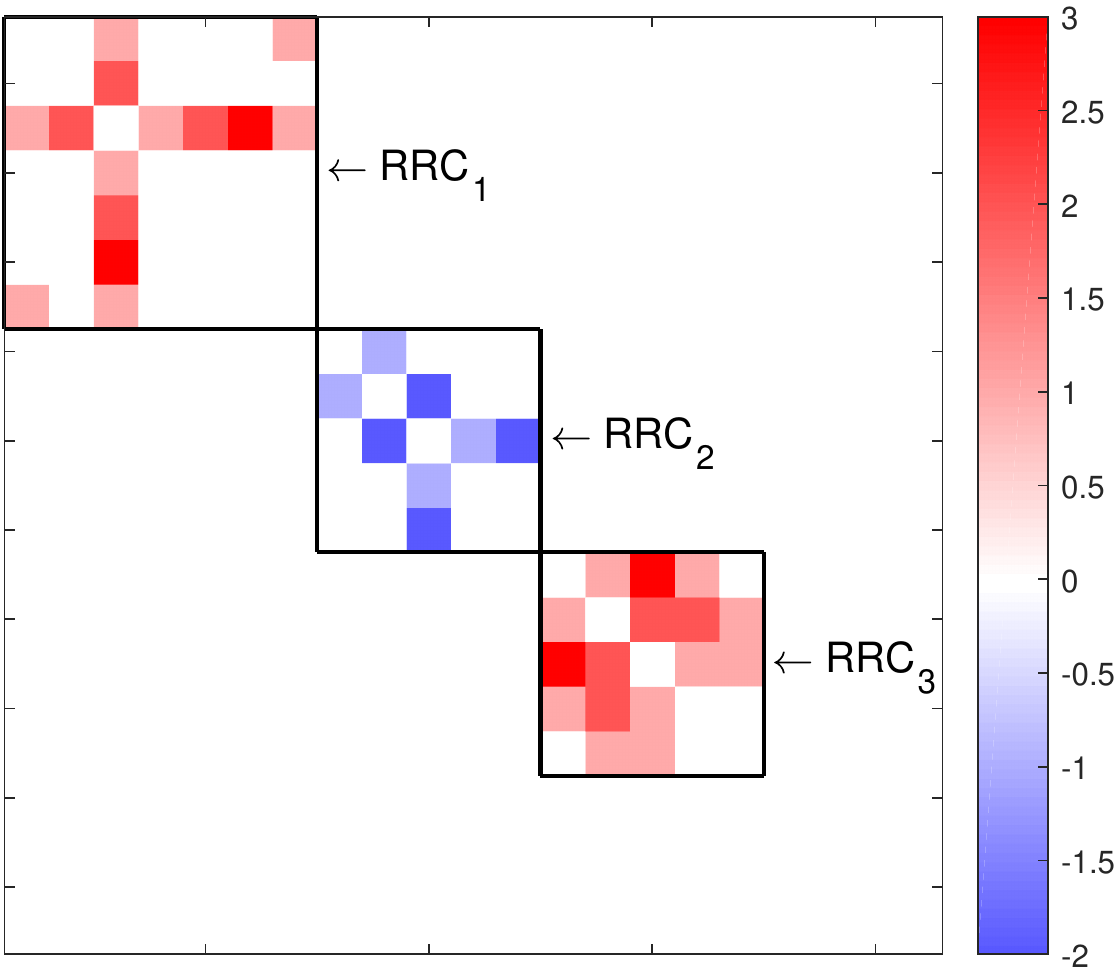}
  \caption{Structure of $B$}
	\label{fig:signalForm}
\end{subfigure}%
\ \
\begin{subfigure}{.195\textwidth}
  \centering
  \includegraphics[width=1\linewidth]{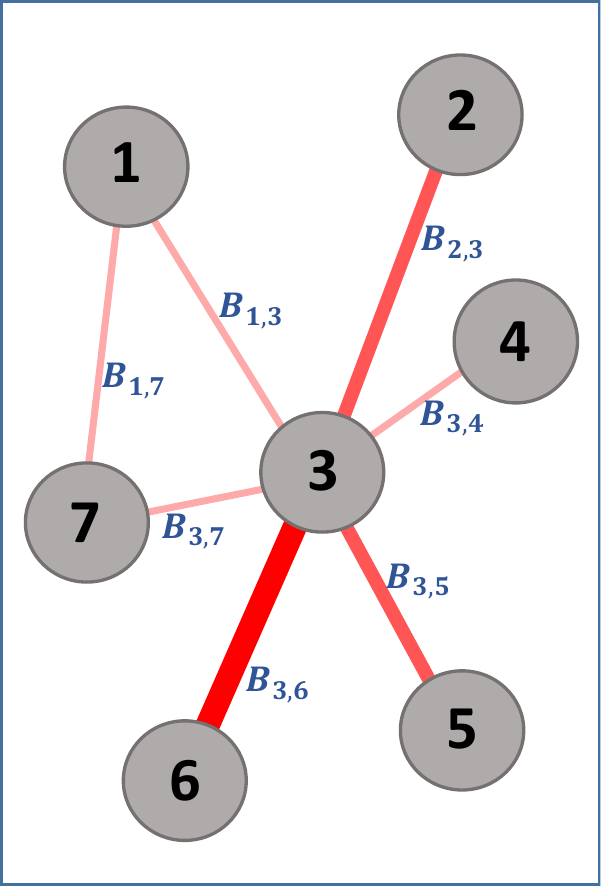}
  \caption{RRC$_1$}
\end{subfigure}
\
\begin{subfigure}{.195\textwidth}
  \centering
  \includegraphics[width=1\linewidth]{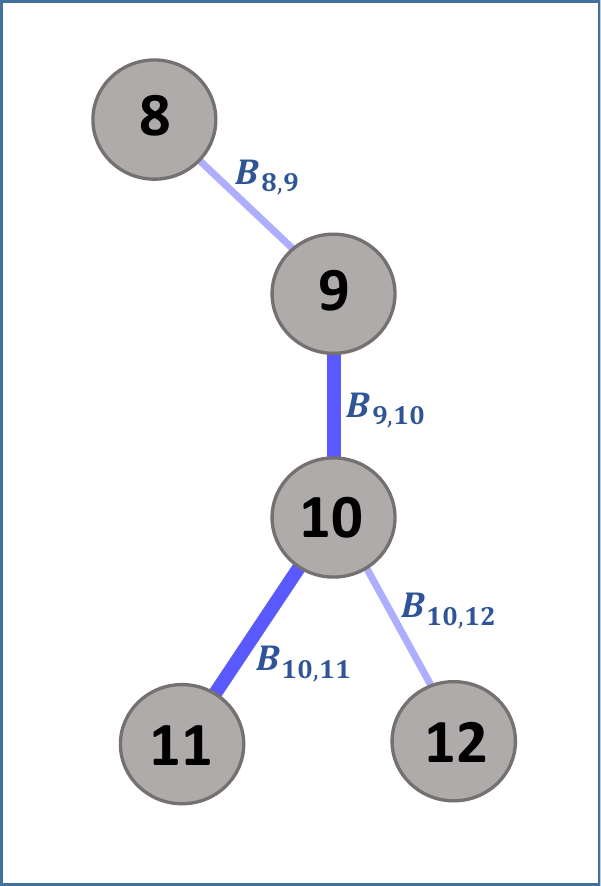}
  \caption{RRC$_2$}
\end{subfigure}
\
\begin{subfigure}{.195\textwidth}
  \centering
  \includegraphics[width=1\linewidth]{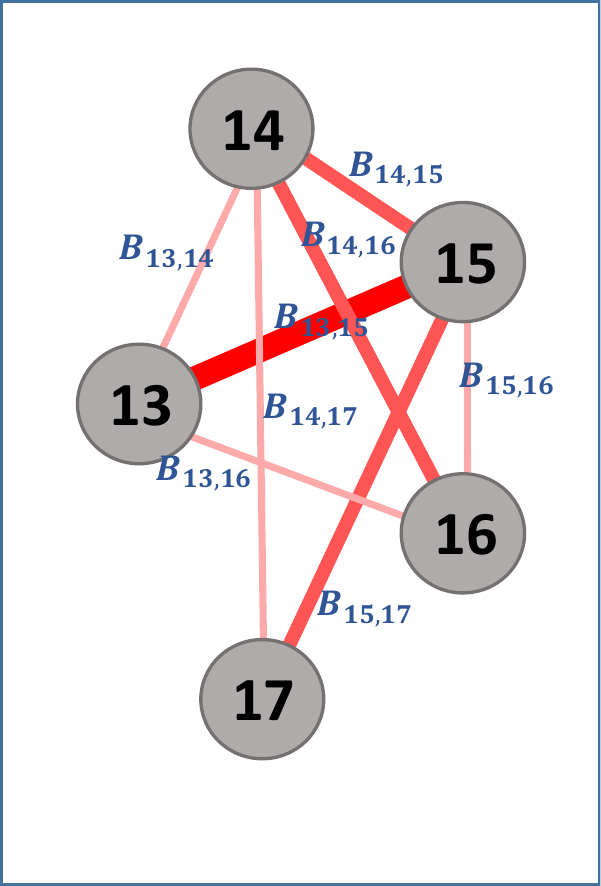}
  \caption{RRC$_3$}
\end{subfigure}
\begin{subfigure}{.32\textwidth}
  \centering
  \includegraphics[width=1\linewidth]{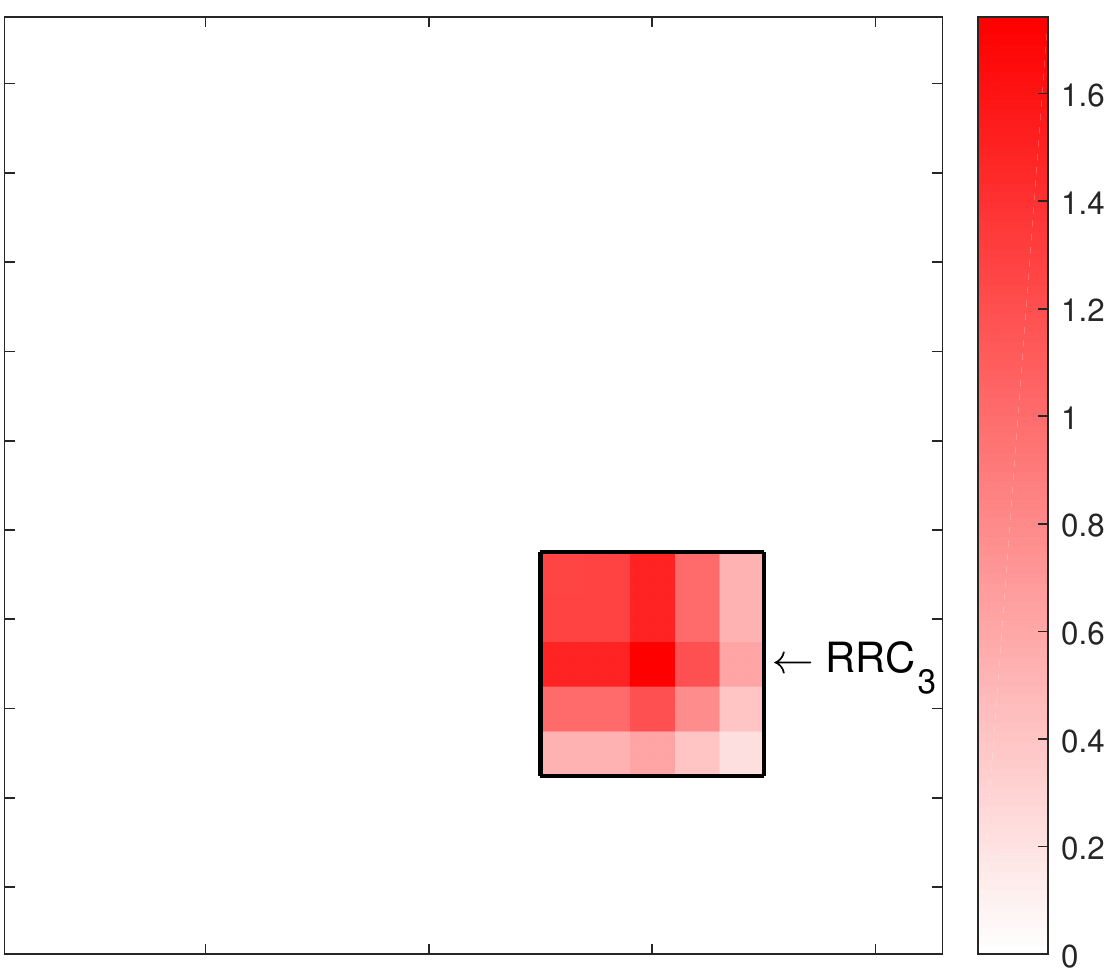}
  \caption{$r_1(B)$}
	\label{fig:1rank}
\end{subfigure}
\begin{subfigure}{.32\textwidth}
  \centering
  \includegraphics[width=1\linewidth]{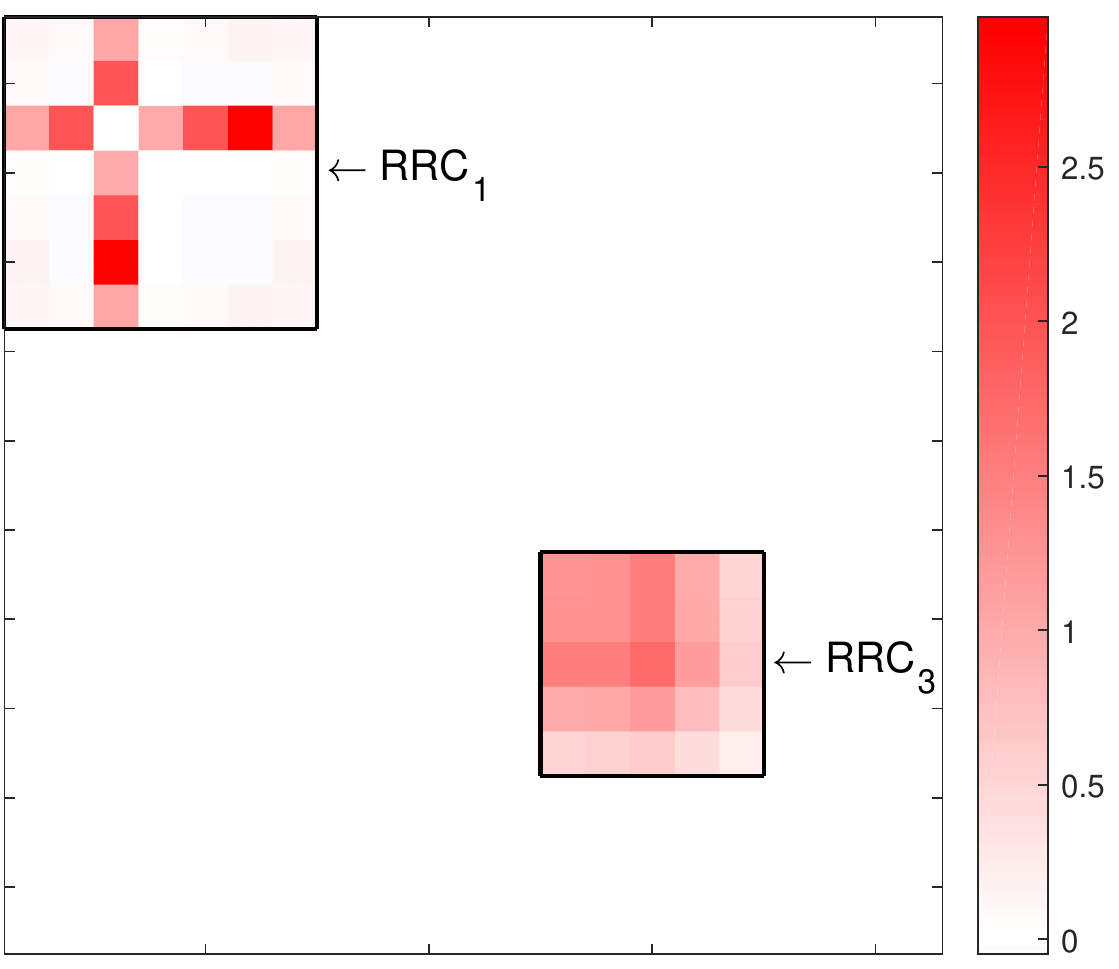}
  \caption{$r_3(B)$}
\end{subfigure}
\begin{subfigure}{.32\textwidth}
  \centering
  \includegraphics[width=1\linewidth]{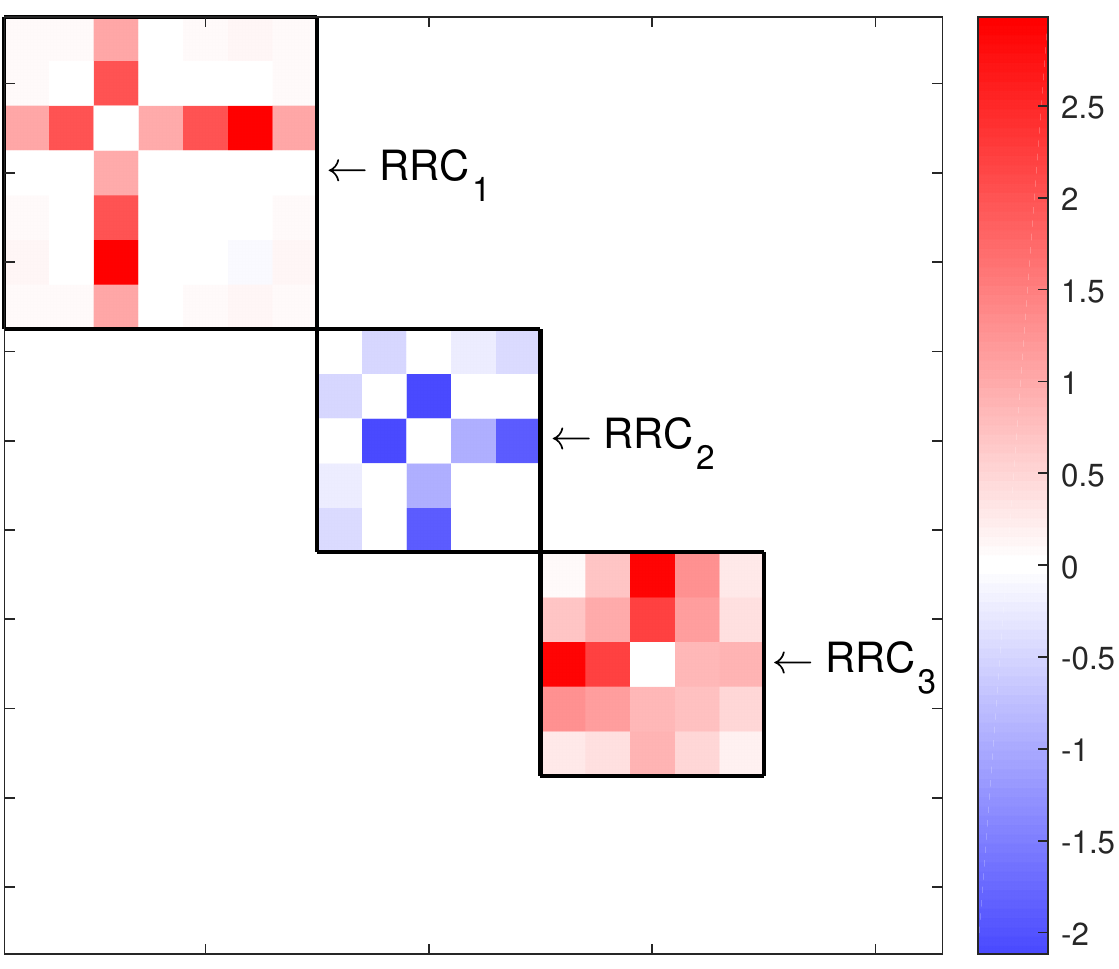}
  \caption{$r_6(B)$}
	\label{fig:6rank}
\end{subfigure}
\caption{The assumed form of $B$ after arranging nodes in the cluster-by-cluster ordering is presented in (a). Three RRC (two positive and one negative) of various connectivity patterns are present. Plots (b)-(d) present the equivalent graph representations of RRCs. Clusters are also easily recognizable in the $k$-rank best (with respect to the Frobenius norm) approximations of $B$, denoted by $r_k(B)$. Only one---the densest---cluster is showed by the 1-rank best approximation (e), while 3-rank and 6-rank approximations, (f) and (g), respectively, can reveal two and three clusters. Panel (g) shows that the low rank approximation of signal may reflect its structure well, although some edges may be lost (as the edge between nodes 1 and 7 in RRC$_1$) or falsely introduced (some edges in RRC$_2$ and RRC$_3$). }
\label{fig:example}
\end{figure}

\subsection{A marriage of sparsity and low rank}


As observed previously, most often we do not have a large enough set of samples to accurately estimate all entries of $B$. More precisely, suppose that we are considering the MLE of $B$ under the model \eqref{model}, without any constraints imposed on the estimates. This leads to the problem of minimizing ${\sum_{i=1}^n\big(y_i -\langle A_i, B \rangle\big)^2}$ with respect to the $p$ by $p$ matrix $B$ (we exclude $X$ and $\beta$ for clarity).
Such a problem does not have a unique solution unless we assume that all $p^2$ vectors $v_{j,k}: = [A_1(j,k),\ldots, A_n(j,k)]\T{T}$ are linearly independent, implying that $n\geq p^2$. If $p = 100$, which is rather a small number of regions compared to brain parcellations widely used in applications, this would necessitate observing data on at least $n=10,000$ subjects.  Of course we can limit the degrees of freedom by assuming (very reasonably) that $B$ is symmetric but this still requires $n>p(p-1)/2=4,950$. We can reduce this number further by assuming there exist relatively few RRCs, so that $B$ is sparse (Figure \ref{fig:signalForm}\,). Let $k$ and $s$ denote the number and the average size of RRCs, respectively. However, even with an oracle telling us precisely the locations of non-zeros and we restrict the estimation to these corresponding entries of $B$) there are still $O(s^2)$ observations required since we have, roughly, $ks^2/2$ entries to estimate (and this is the simplest scenario with all clusters having the same number of nodes).

Again consider a signal with a block pattern as in Figure \ref{fig:signalForm}. For such matrices, the first a few eigenvectors (i.e., those corresponding to the largest eigenvalues) are of a special form. Assuming that they are unique (up to a change in sign), each of them may have its non-zeros located inside exactly one RRC. In our example (presented in Figure \ref{fig:example}), the eigenvector $v_1$, corresponding to the largest eigenvalue $\lambda_1$, has all its non-zeros located inside RRC$_3$. Moreover, as a direct consequence of Perron-Frobenius theorem \citep{FrobTheorem} for irreducible matrices,  all coefficients of $v_1$ located inside RRC$_3$ must be either {\textsl{strictly positive} or \textsl{strictly negative}. That is, the non-zero entries of $v_1$ coincide precisely with the indices of RRC$_3$ and the matrix $\lambda_1v_1v_1\T{T}$ (i.e., the best rank-one approximation of $B$) and corresponds to the entire corresponding block in Figure \ref{fig:1rank}. In fact, the first several, say $\tilde{k} \ll p$, eigenvectors may effectively summarize the structure of signal via the best rank-$\tilde{k}$ approximation, $r_{\tilde{k}}(B): = \sum_{i=1}^{\tilde{k}}\lambda_iv_iv_i\T{T}$, of $B$ (Figure \ref{fig:6rank}).

Returning to the calculations, if we restrict attention to the general structure of $B$ reflected by its first $\tilde{k}$ eigenvectors and assume that each of them has roughly $s$ non-zeros within one of the RRCs, then we obtain $O(s)$, namely $\tilde{k}s$, coefficients to estimate (according to an oracle). This refocus on ``structured sparsity" not only reduces the computational requirements, but also adds to the physiological interpretation. 

In summary, we propose a method called \textit{SParsity Inducing Nuclear Norm EstimatoR} ($\SP$) that constructs a low-rank and sparse estimate of $B$ under the model \eqref{model}. It provides a principled approach to estimating $B$ via: (i) exploiting its dominant eigenvectors to accurately estimate block-structured coefficients and, simultaneously, (ii) imposing sparsity outside of these blocks. 

\section{Methodology}
\label{sec_Method}
\subsection{Penalized optimization}
With the goal of encouraging a regression coefficient (matrix) estimate to be both sparse and low-rank, $\SP$ employs two types of matrix norms which cooperate together as penalties to regularize the estimate: an $\ell_1$ norm imposes entry-wise sparsity and a {\it nuclear norm} achieves a convex relaxation of rank minimization (\cite{Candes2009, Rech2010}).
The nuclear norm (also referred to as the trace norm) of $B$, denoted by $\|B\|_*$, is defined as a sum of the singular values of $B$.

For a pair of prespecified nonnegative tuning parameters $\lambda_N$ and $\lambda_L$, $\SP$ is defined as a solution to the following optimization problem
\begin{equation}
\label{SPINNER}
\big\{\hat{B}^{\ES{S}},\ \hat{\beta}^{\ES{S}} \big\}:=\argmin {B, \beta}\ \left\{\frac12\sum_{i=1}^n\Big(y_i -\langle A_i, B \rangle - X_i\beta \Big)^2+ \lambda_N \big\|B\big\|_* + \lambda_L\big\|\vvec(W\circ B)\big\|_1\right\},
\end{equation}
where $W$ is a $p\times p$ symmetric matrix of nonnegative weights. Here, $W\circ B$ denotes the Hadamard product (entrywise product) of $W$ and $B$, hence $\big\|\vvec(W\circ B)\big\|_1 = \sum_{j,l = 1}^pW_{j,l}|B_{j,l}|$. By default, $W$ is the matrix with zeros on the diagonal and ones on the off-diagonal. Setting all $W_{j,j}$s as zeros protects the diagonal entries of $\hat{B}^{\ES{S}}$ from being shrunk to zero by the $\ell_1$ norm and leads to a more accurate recovery of a low-rank approximation of $B$ via the nuclear norm. Penalizing $\hat{B}^{\ES{S}}_{j,j}$s by $\ell_1$ norm in a situation when all $A_i$s have zeros on their diagonals would not be justified, since the diagonal $B_{j,j}$s (the nodes' effects) are not included in model \eqref{model}; i.e., there is no information about them in $y$. 
We note that with the default $W$,  $B\mapsto \big\|\vvec(W\circ B)\big\|_1$ does not define norm, however it is a convex function (and a seminorm) so \eqref{SPINNER} is a convex optimization problem.

$\SP$ is well-defined in the sense that a solution to \eqref{SPINNER} exists for any tuning parameters $\lambda_N, \lambda_L$  (see the subsection \ref{subs:basic}).  If $\lambda_N = 0$, \eqref{SPINNER} has infinitely many pairs $\big\{\hat{B}^{\ES{S}},\ \hat{\beta}^{\ES{S}} \big\}$ minimizing \eqref{SPINNER} with a default selection of weights, since the diagonal entries of $\hat{B}^{\ES{S}}$ do not impact the objective function. The off-diagonal elements of $\hat{B}^{\ES{S}}$ (which in that case reduces to a lasso estimate) will, however, be unique with probability one if the predictor variables are assumed to be drawn from a continuous probability distribution \citep[see,][]{tibshirani2012lasso}. In a situation of non-uniqueness, the name ``$\SP$'' will refer to the set of all solutions to \eqref{SPINNER}.

\subsection{Simplifying the optimization problem}
\label{simpOpt}

The problem in \eqref{SPINNER} is defined as an optimization with respect to both $B$ and $\beta$, but this can be reformulated so that, in practice, we need only solve a minimization problem with respect to $B$.  To see this, define the vector $w_B$ as $(w_B)_i: = y_i -\langle A_i, B \rangle$ and note that
\begin{equation*}
\hat{\beta}_{B}  := \argmin{\beta} \ \frac12\sum_{i=1}^n\Big( y_i-\langle A_i, B \rangle - X_i\beta \Big)^2  \\
                 = \argmin{\beta} \big\| w_b - X\beta\big\|^2 = \big(X\T{T}X\big)^{-1}X\T{T}w_B.
\end{equation*}
Since the penalty terms involving $B$ can be treated as additive constants, this solves \eqref{SPINNER} with respect to $\beta$.  Therefore, we can substitute $\hat{\beta}_{B}$ into \eqref{SPINNER} and transform the problem into one involving only $B$. For this, denote the projection onto the orthogonal complement of the range of $X$ as $H: = \mathbf{I}_n - X\big(X\T{T}X\big)^{-1}X\T{T}$. Also, denote by $\mathcal{A}$ the $n$-row matrix of stacked vectors from $\{\vvec(A_i)\}_{i=1}^n$. If we transform $y$ and $\mathcal{A}$ as $\widetilde{y}: = Hy$ and $\widetilde{\mathcal{A}}: = H\mathcal{A}$, then upon substitution of $\hat{\beta}_{B}$ into \eqref{SPINNER}, we can rewrite the model-fit term as
\begin{equation*}
\begin{split}
&\sum_{i=1}^n\Big(y_i -\langle A_i, B \rangle - X_i\hat{\beta}_B \Big)^2 =\sum_{i=1}^n\Big(y_i -\vvec(A_i)\T{T}\vvec(B) - X_i\hat{\beta}_B \Big)^2
 =\big\|y - \mathcal{A}\vvec(B) - X\hat{\beta}_B\big\|_2^2= \\
 &\Big\|y - \mathcal{A}\vvec(B) - X\big(X\T{T}X\big)^{-1}X\T{T}\big(y - \mathcal{A}\vvec(B)\big)\Big\|_2^2
 =\big\|Hy - H\mathcal{A}\vvec(B)\big\|_2^2 =\big\|\widetilde{y} - \widetilde{\mathcal{A}}\vvec(B)\big\|_2^2.
\end{split}
\end{equation*}
Hence, \eqref{SPINNER} can be equivalently represented as
\begin{equation}
\label{SPINNER3}
\left\{
\begin{array}{l}
\hat{B}^{\ES{S}}:=\argmin {B}\ \bigg\{ \big\|\widetilde{y} - \widetilde{\mathcal{A}}\vvec(B)\big\|_2^2  + \lambda_N \big\|B\big\|_* + \lambda_L\big\|\vvec(W\circ B)\big\|_1\bigg\}\\
\hat{\beta}^{\ES{S}}: = \big(X\T{T}X\big)^{-1}X\T{T}\big[y - \mathcal{A}\vvec(\hat{B}^{\ES{S}})\big]
\end{array}
\right..
\end{equation}

\subsection{Basic properties}
\label{subs:basic}
In view of the reformulation of $\SP$ in \eqref{SPINNER3} we will, without loss of generality, exclude $X$ and $\beta$ from consideration and focus on the minimization problem with the objective function
\begin{equation}
\label{objective}
F(B):\ =\ \underbrace{\frac12\sum_{i=1}^n\Big(y_i -\langle A_i, B \rangle \Big)^2}_{f(B)}\ +\ \underbrace{\lambda_N \big\|B\big\|_*}_{g(B)}\ +\ \underbrace{\lambda_L\big\|\vvec(W\circ B)\big\|_1}_{h(B)}.
\end{equation}
The following proposition clarifies that a $\SP$ estimate is well-defined.
\begin{proposition}
\label{prop0}
For any pair of regularization parameters $\lambda_N\geq 0$ and $\lambda_L\geq 0$ there exists at least one solution to \eqref{objective}. The claim is still valid, if the function $g(B) + h(B)$ in \eqref{objective} is replaced by any nonnegative convex function $\tilde{g}$.
\end{proposition}
\begin{proof}
We use the concept of \textit{directions of recession} \citep{Rockafellar1970}. In our situation, a matrix $C$ belongs to the set of directions of recession of $F$ if $F(B + \lambda C) \leq F(B)$ for any matrix $B$ and any scalar $\lambda \geq 0$.
In particular, for $B=0$, the direction of recession of $F$ must satisfy $F(\lambda C) - F(0) \leq 0$ for any $\lambda \geq 0$. Therefore,
\begin{equation}
\label{dirrec}
\frac12\sum_{i=1}^n\Big[\big(y_i -\lambda\langle A_i, C \rangle \big)^2 - y_i^2\Big]\ +\ \lambda\cdot\lambda_N \big\|C\big\|_*\ +\ \lambda\cdot\lambda_L\big\|\vvec(W\circ C)\big\|_1 \leq 0.
\end{equation}
Since the last two terms in \eqref{dirrec} are nonnegative, it also holds that $\sum_{i=1}^n\big[\big(y_i -\lambda\langle A_i, C \rangle \big)^2 - y_i^2\big] \leq 0$, hence $\lambda^2 \sum_{i=1}^n \langle A_i, C \rangle^2 - 2\lambda \sum_{i=1}^n y_i\langle A_i, C \rangle \leq 0$, for all $\lambda\geq 0 $. This can happen only when $\sum_{i=1}^n \langle A_i, C \rangle^2 = 0$, implying that $\langle A_i, C \rangle = 0$ for each $i$. Combining this with \eqref{dirrec} gives also $\lambda_N \big\|C\big\|_* = 0$ and $\lambda_L\big\|\vvec(W\circ C)\big\|_1 = 0$. When $\lambda_N>0$ this implies $C=0$,  but any selection of regularization parameters imply that $C$ must be a direction in which the objective function is constant. Therefore, applying Theorem 27.1(b) from \cite{Rockafellar1970}, $F$ attains its minimum.
\end{proof}

\noindent Since $B$ is assumed to be symmetric, it is natural to expect estimates of $B$ to have the same property, although, as yet, we have not enforced this condition on $\hat{B}^{\ES{S}}$. Fortunately, as shown next, we can always obtain a symmetric minimizer of $F$ in \eqref{objective}.
\begin{proposition}
\label{prop1}
Suppose that $W$ and all matrices $A_i$s are symmetric. Then, the set of solutions to the minimization problem with an objective function defined in \eqref{objective} contains a symmetric matrix. The claim is still valid if the function $g(B) + h(B)$ in \eqref{objective} is replaced by any nonnegative convex function $\tilde{g}$ such as $\tilde{g}(A\T{T}) = \tilde{g}(A)$ for any $p\times p$ matrix $A$.
\end{proposition}
\begin{proof}
From Proposition \ref{prop0} we know that there exists a solution $B^* =\argmin{B} F(B)$. Consider its symmetric part, $\widetilde{B}: = \frac12( B^* + {B^*}\T{T})$. By the symmetry of each $A_i$, for $f$ defined in \eqref{objective}
\begin{equation*}
f(\widetilde{B}) = \frac12\sum_{i=1}^n\Big(y_i -\frac12\langle A_i, B^* \rangle - \frac12\langle A_i, {B^*}\T{T} \rangle \Big)^2 =\frac12\sum_{i=1}^n\Big(y_i -\frac12\langle A_i, B^* \rangle - \frac12\langle A_i\T{T}, B^*\rangle \Big)^2 = f(B^*).
\end{equation*}
Now,
\begin{equation*}
\begin{split}
g(\widetilde{B}) + h(\widetilde{B}) &\ = \ \lambda_N \Big\| \frac12B^* + \frac12{B^*}\T{T}\Big\|_*\ +\ \lambda_L\Big\|\frac12\vvec( W\circ B^*) + \frac12\vvec(W\circ{B^*}\T{T})\Big\|_1\\
&\ \leq \ \frac{\lambda_N}2 \big\| B^* \big\|_*\, +\, \frac{\lambda_N}2 \big\|{B^*}\T{T}\big\|_*\ +\ \frac{\lambda_L}2\big\|\vvec( W\circ B^*)\big\|_1\,+\,\frac{\lambda_L}2\big\|\vvec(W\circ{B^*}\T{T})\big\|_1\\
&\ = \ \lambda_N \big\| B^* \big\|_* + \lambda_L\big\|\vvec(W\circ B^*)\big\|_1\ =\ g(B^*) + h(B^*),
\end{split}
\end{equation*}
where the inequality follows from the fact that $g$ and $h$ are convex and the last equality holds since both these functions are invariant under transpose, provided that $W$ is symmetric. Consequently, we get $F(\widetilde{B})\leq F(B^*)$, hence $\widetilde{B}$ must be a solution.
\end{proof}
\noindent In summary, this shows that the symmetric part of any solution to \eqref{objective} is also a solution. In particular, when the solution is unique, it is guaranteed to be a symmetric matrix.

\begin{proposition}
\label{permProp}
Let $\hat{B}$ be a solution to minimization problem with an objective, $F(B)$, defined in \eqref{objective}. We consider the modification of the data relying on the nodes reordering. Precisely, suppose that $\pi:\{1,\ldots, p\} \rightarrow \{1,\ldots, p\}$ is a given permutation with corresponding permutation matrix $P_{\pi}$, i.e. it holds $P_{\pi}v = [v_{\pi(1)},\ldots, v_{\pi(p)}]\T{T}$ for any column vector $v$. We replace the matrices $A_i$'s and $W$ in \eqref{objective} with matrices having rows and columns permuted by $\pi$, namely,\ $A^{\pi}_i: = P_{\pi}A_iP_{\pi}\T{T}$ \ and \ $W^{\pi}: = P_{\pi}WP_{\pi}\T{T}$. Then, $\hat{B}$ with rows and columns permuted by $\pi$, i.e. $\hat{B}^{\pi}: = P_{\pi}\hat{B}P_{\pi}\T{T}$, is a solution to the updated problem.
\end{proposition}
\begin{proof}
Suppose that $\hat{B}^{\pi}$ is not a solution, hence there exists matrix $C$ such as 
\begin{align}
\label{21012020}
\begin{split}
\frac12\sum\limits_{i=1}^n\big(y_i -\langle A^{\pi}_i,\, C \rangle \big)^2\ +\ \lambda_N & \big\|C\big\|_*\ +\ \lambda_L\big\|\vvec(W^{\pi}\circ C)\big\|_1 < \\[-1ex]
&\frac12\sum\limits_{i=1}^n\big(y_i -\langle A^{\pi}_i,\, \hat{B}^{\pi} \rangle \big)^2\ +\ \lambda_N \big\|\hat{B}^{\pi}\big\|_*\ +\ \lambda_L\big\|\vvec(W^{\pi}\circ \hat{B}^{\pi})\big\|_1.
\end{split}
\end{align}
We have $\langle A^{\pi}_i, C \rangle = \langle P_{\pi}A_iP_{\pi}\T{T}, C \rangle = \operatorname{tr}\big(P_{\pi}A_iP_{\pi}\T{T}C\big) = \operatorname{tr}\big(A_iP_{\pi}\T{T}CP_{\pi}\big) = \langle A_i, P_{\pi}\T{T}CP_{\pi} \rangle = \langle A_i, \widetilde{C} \rangle$, for $\widetilde{C}:=P_{\pi}\T{T}CP_{\pi}$. Moreover, $\big\|\vvec(W^{\pi}\circ C)\big\|_1 = \big\|\vvec(P_{\pi}WP_{\pi}\T{T}\circ C)\big\|_1 = \big\|\vvec(P_{\pi}WP_{\pi}\T{T}\circ P_{\pi}\widetilde{C}P_{\pi}\T{T})\big\|_1 = \big\|\vvec\big(P_{\pi}(W\circ \widetilde{C})P_{\pi}\T{T}\big)\big\|_1 = \sum\limits_{j,l}\big|W_{\pi(j),\pi(l)}\widetilde{C}_{\pi(j),\pi(l)}\big|= \sum\limits_{j,l}\big|W_{j,l}\widetilde{C}_{j,l}\big|=\big\|\vvec(W\circ \widetilde{C})\big\|_1$, where the third equation follows from the exchangeability of Hadamard product and permutation imposed on rows or columns of matrices (provided that the same permutation is used for two matrices). Since $C$ and $\widetilde{C}$ share the same singular values, it also holds $\|C\|_* = \|\widetilde{C}\|_*$. Consequently, the left-hand side of \eqref{21012020} can be simply expressed as $F(\widetilde{C})$.

On the other hand we have $\langle A^{\pi}_i,\, \hat{B}^{\pi} \rangle = \operatorname{tr}\big( P_{\pi}A_iP_{\pi}\T{T} P_{\pi}\hat{B}P_{\pi}\T{T}\big)= \operatorname{tr}\big(A_i\hat{B}\big) = \langle A_i,\, \hat{B} \rangle$, since $P_{\pi}\T{T} P_{\pi} = \mathbf{I}$. As above, we can get rid of the permutation symbols inside the nuclear and $\ell_1$ norms, yielding $\|\hat{B}^{\pi}\|_* = \|\hat{B}\|_*$ and $\big\|\vvec(W^{\pi}\circ \hat{B}^{\pi})\big\|_1 = \big\|\vvec(W\circ \hat{B})\big\|_1$. Therefore, the right-hand side of \eqref{21012020} becomes $F(\hat{B})$ and the inequality yields $F(\widetilde{C}) < F(\hat{B})$ which contradicts the optimality of $\hat{B}$ and proves the claim.
\end{proof}
\noindent The above statement implies that $\SP$ is invariant under the order of nodes in a sense that the rearrangement of the nodes simply corresponds to the rearrangement of rows and columns of an estimate. Consequently, there is no need for fitting the model again. More importantly, the \textit{optimal order} of nodes, i.e. the permutation which reveals the assumed clumps structure (see, Section \ref{sec_BrainDataResults}), can be found at the end of the procedure based on the $\SP$ estimate achieved for \textit{any} arrangement of nodes, e.g. corresponding to the alphabetical order of node labels. 

\subsection{Numerical implementation}\label{sec:numSol}
To build the numerical solver for the problem \eqref{objective}, we employed the Alternating Direction Method of Multipliers (ADMM) \citep{Gabay1976ADA}. The algorithm relies on introducing $p\times p$ matrices $C$ and $D$ as new variables and considering the constrained version of the problem (equivalent to \eqref{objective}) with a separable objective function:
\begin{equation}
\arraycolsep=1.4pt\def\arraystretch{0.8}
\argmin{B, C, D}\  \big\{f(B)\,+\,g(C)\,+\,h(D)\big\} \qquad \textrm{such that }\ \left\{\begin{array}{l}
D-B=0\\
D-C=0
\end{array}
\right..
\end{equation}
The augmented Lagrangian with the scalars $\delta_1>0$, $\delta_2>0$ and dual variable $\dual: = \left [\begin{BMAT}(c)[0.5pt,0pt,0.7cm]{c}{cc}\dual_1 \\ \dual_2 \end{BMAT} \right ]\in \mathbb{R}^{2p\times p}$ is
\begin{equation*}
L_{\delta}(B,C,D; \dual) =f(B)+g(C)+h(D)+\langle\dual_1,D-B\rangle+\langle\dual_2,D-C\rangle+\frac{\delta_1}2\big\|D-B\big\|_F^2+\frac{\delta_2}2\big\|D-C\big\|_F^2.
\end{equation*}

ADMM builds the update of the current guess, i.e., the matrices $B^{[k+1]}$, $C^{[k+1]}$ and $D^{[k+1]}$, by  minimizing $L_{\delta}(B,C,D; \dual)$ with respect to each of the primal optimization variables separately while treating all remaining variables as fixed. Dual variables are updated in the last step of this iterative procedure. Since $\langle\dual_1,D-B\rangle+\frac{\delta_1}2\|D-B\|_F^2 =\frac{\delta_1}2\|D+\frac{\dual_1}{\delta_1}-B\|_F^2 + const_1$ and $\langle\dual_2,D-C\rangle+\frac{\delta_2}2\|D-C\|_F^2 =\frac{\delta_2}2\|D+\frac{\dual_2}{\delta_2}-C\|_F^2 + const_2$, where $const_1$ does not depend on $B$ and $const_2$ does not depend on $C$, ADMM updates for the considered problem takes the final form
\begin{align}
&B^{[k+1]}:=\ \argmin {B}\bigg\{\,2f(B)\ +\ \delta^{[k]}_1\Big\|\,D^{[k]}  + \frac{\dual_1^{[k]}}{\delta^{[k]}_1}- B\,\Big\|_F^2\,\bigg\}\label{Update1}\\
&C^{[k+1]}:=\ \argmin {C}\bigg\{\,2g(C)\ +\ \delta^{[k]}_2\Big\|\,D^{[k]}  + \frac{\dual_2^{[k]}}{\delta^{[k]}_2}- C\,\Big\|_F^2\,\bigg\}\label{Update2}\\
&D^{[k+1]}:=\ \argmin {D}\bigg\{\,2h(D)\ +\ \delta^{[k]}_1\Big\|\,D + \frac{\dual_1^{[k]}}{\delta^{[k]}_1} - B^{[k+1]}\,\Big\|_F^2\ +\ \delta^{[k]}_2\Big\|D + \frac{\dual_2^{[k]}}{\delta^{[k]}_2} -C^{[k+1]}\Big\|_F^2\,\bigg\}\label{Update3}\\
&\left\{
\begin{array}{l}
\dual_1^{[k+1]}: =\ \dual_1^{[k]} + \delta^{[k]}_1\big(D^{[k+1]}-B^{[k+1]}\big)\\
\dual_2^{[k+1]}: =\ \dual_2^{[k]} + \delta^{[k]}_2\big(D^{[k+1]}-C^{[k+1]}\big)
\end{array}
\right..
\end{align}

All of the subproblems \eqref{Update1}, \eqref{Update2} and \eqref{Update3} have analytical solutions and can be computed very efficiently (see Section \ref{subs:subproblems} in the Appendix). Here, the positive numbers $\delta^{[k]}_1$ and $\delta^{[k]}_2$ are treated as the step sizes. The convergence of ADMM is guaranteed under very general assumptions when these parameters are held constant. However, their selection should be performed with caution since they strongly impact the practical performance of ADMM \citep{Xu2017AdaptiveRA}. Our \texttt{MATLAB} implementation uses the procedure based on the concept of \textit{residual balancing} \citep{Wohlberg2017ADMMPP, Xu:2017} in order to automatically modify the step sizes in consecutive iterations and provide fast convergence. The stopping criteria are defined as the simultaneous fulfilment of the conditions
$$\max\left\{\frac{\|C^{[k+1]}-B^{[k+1]}\|_F}{\|B^{[k+1]}\|_F}, \frac{\|D^{[k+1]}-B^{[k+1]}\|_F}{\|B^{[k+1]}\|_F}\right\}< \epsilon_P,\qquad \frac{\|D^{[k+1]}-D^{[k]}\|_F}{\|D^{[k]}\|_F}< \epsilon_D,$$
which we use as a measure stating that the primal and dual residuals are sufficiently small \citep{Wohlberg2017ADMMPP}. The default settings are $\epsilon_P: = 10^{-6}$ and  $\epsilon_D: = 10^{-6}$.

\section{Simulation Experiments}
\label{sec_SimulationExperiments}
We now investigate the performance of the proposed method, $\SP$, and compare it with nuclear-norm regression \citep{Zhou-regularized-2014}, lasso \citep{LassoF}, elastic net \citep{elasticNet} and ridge \citep{ridge}. Without loss of generality we focus the simulations on the model where there are no additional covariates, i.e. ${y_i = \langle A_i, B \rangle + \varepsilon_i}$, for $i=1,\ldots, n$ and
with $\varepsilon\sim\mathcal{N}(0,\sigma^2I_n)$.

\subsection{Considered scenarios}
We consider three scenarios. In the first two scenarios, the observed matrices $\{A_i\}_{i=1}^n$ are synthetic and the ``true" signal is defined by a pre-specified $B$: Scenario 1 considers the effects of signal strength (determined by $B$) and Scenario 2 studies power and the effects of sample size, $n$. In Scenario 3, the $A_i$'s are from real brain connectivity maps. Specifically:
\begin{enumerate}[label=Scenario \arabic*, ref=Scenario \arabic*, leftmargin=1.7cm]
\item\label{S1} For each $A_i$, its upper triangular entries are first sampled independently from $\mathcal{N}(0,1)$. Then these entries are standardized element-wise across $i$ to have mean 0 and standard deviation 1. The lower triangular entries are obtained by symmetry and the diagonal entries are set at 0. $B$ is block-diagonal $\{\mathbbm{1}_{8\times8}, -s\times\mathbbm{1}_{8\times8}, s\times\mathbbm{1}_{8\times8}, \mathbf{0}_{(p-24)\times(p-24)}\}$, where $p = 60$, $\mathbbm{1}_{8\times8}$ denotes an $8\times8$ matrix of ones, and $s = 2^k$ with $k\in \{-3,-2,\ldots,5\}$. The noise level is $\sigma=0.1$ and the number of observations is $n = 150$.
\item\label{S2} In this scenario, the $A_i$'s are obtained in the same manner as in \ref{S1} except that we fix $s = 1$ and vary $n \in \{50, 100, \ldots, 300\}$.
\item\label{S3} The $A_i$'s are real functional connectivity matrices of 100 unrelated individuals from the Human Connectome Project (HCP) \citep{conn}. Data were preprocessed in FreeSurfer \citep{FreeS}. We removed the subcortical areas what resulted in a final number of $p = 148$ brain regions. As before, entries in $A_i$'s are standardized element-wise across $i$ before $y$ is generated. $B$ is block-diagonal $\{\mathbf{0}_{56\times56}, \mathbbm{1}_{6\times6}, \mathbf{0}_{6\times6}, -s\times\mathbbm{1}_{5\times5}, \mathbf{0}_{49\times49}, s\times\mathbbm{1}_{8\times8}, \mathbf{0}_{18\times18}\}$, $s = 2^k$ with $k\in \{-3,-2,\ldots,5\}$, and $\sigma = 0.1$.
\end{enumerate}
For each setting in \ref{S1} and \ref{S2}, the process of generating $A$ and $y$ is repeated $100$ times. For each setting in \ref{S3}, the process of generating $y$ is repeated $100$ times.

\subsection{Simulation implementation}\label{subsec_SimulationImplementation}
For each simulation setting, we apply the following five regularization methods to estimate the matrix $B$: $\SP$, elastic net, nuclear-norm regression, lasso, and ridge. $\SP$ and elastic net both involve penalizing two types of norms, while the others use a single norm for the penalty term. The regularization parameters for these methods are chosen by five-fold cross-validation, where the fold membership of each observation is the same across methods.

For $\SP$, we consider a $15\times 15$ two-dimensional grid of paired parameter values $(\lambda_L, \lambda_N)$. The smallest value  for each coordinate is zero and the largest is the smallest value that produces $\hat{B}=0$ when the other coordinate is zero. The other 13 values for each coordinate are equally spaced on a logarithmic scale. The optimal $(\lambda_L^*, \lambda_N^*)$ is chosen as the pair that minimizes the average cross-validated squared prediction error. Elastic net regression requires the selection of two tuning parameters, $\alpha$ and $\lambda$.  Their selection is implemented using the \texttt{MATLAB} package \texttt{glmnet} \citep{glmnet}, in which we consider 15 equally-spaced (between 0 and 1) $\alpha$ values, where $\alpha = 0$ corresponds to ridge regression and $\alpha = 1$ corresponds to lasso regression. Then, for each $\alpha$, we let \texttt{glmnet} optimize over 15 automatically chosen $\lambda$ values, and pick the one that minimizes the average cross-validated squared prediction error.
We treat nuclear-norm  and lasso regression as special cases of $\SP$ where one of the regularization parameters is set at 0, and the other one takes on the same 15 values as specified for the $\SP$. The simplified versions of ADMM algorithm for these two special cases can be found in Appendix \ref{app:degenerate}}. Ridge regression is implemented using \texttt{glmnet} where the optimization is done over 15 automatically chosen $\lambda$ values.

The default form of $W$ is used for $\SP$ so that the diagonal elements of $B$ are not penalized. Since the connectivity matrices $A_i$ have zeros on their diagonals, the off-diagonal elements for the elastic net, lasso, and ridge regression are the same regardless of whether the diagonal is penalized and therefore we do not exclude diagonal elements from penalization (consequently, they are always estimated as zeros for these methods). As will be discussed in the following section, diagonal elements do not enter our evaluation criterion.

\subsection{Simulation Results}
We measure the performance of each method by the relative mean squared error between its estimator $\hat{B}$ and the true $B$, defined as $\text{MSEr} = \Vert \hat{B} - B\Vert_{F^{\star}}^2/\Vert B\Vert_{F^{\star}}^2$, where $\Vert\cdot\Vert_{F^{\star}}$ denotes the Frobenius norm of a matrix, excluding diagonal entries. In other words, $\Vert B\Vert_{F^{\star}}^2$ is the sum of squared off-diagonal entries of $B$.

\subsubsection{Scenario 1: Synthetic Connectivity Matrices with Varying Signal Strengths}\label{ssec:scenario1}
Figure~\ref{fig_Scenario1} shows the relative mean squared errors of $\hat B$ for the five regularization methods as $\log_2(s) \in \{-3,-2,\ldots,5\}$ under \ref{S1}, where the nonzero entries in $B$ consist of blocks $\mathbbm{1}_{8\times8}$, $-s\times\mathbbm{1}_{8\times8}$, and $s\times\mathbbm{1}_{8\times8}$. We can see that $\SP$ outperforms the other methods for all values of $s$ and produces a relative mean squared error much smaller than that of elastic net, lasso, and ridge. It is also observed that for $\SP$, nuclear-norm regression, elastic net, and lasso, their relative mean squared errors are at the highest when $s = 1$. This can be explained as follows. When $s\ll 1$, the block $\mathbbm{1}_{8\times8}$ dominates the blocks $-s\times\mathbbm{1}_{8\times8}$ and $s\times\mathbbm{1}_{8\times8}$, making them more like noise terms so that effectively the number of response-relevant variables closer to 64, which is smaller than the number of observations $n = 150$. Similarly, when $s\gg 1$, the blocks $-s\times\mathbbm{1}_{8\times8}$ and $s\times\mathbbm{1}_{8\times8}$ dominate the block $\mathbbm{1}_{8\times8}$, effectively making the number of variables closer to 128, which is still smaller than 150. However, when $s = 1$, the total number of response-relevant variables is 192, which is larger than the number of observations, making the estimation of $B$ in this case more difficult.

\begin{figure}[ht]
\centering
\begin{subfigure}{.325\textwidth}
  \centering
	\includegraphics[width=1\linewidth]{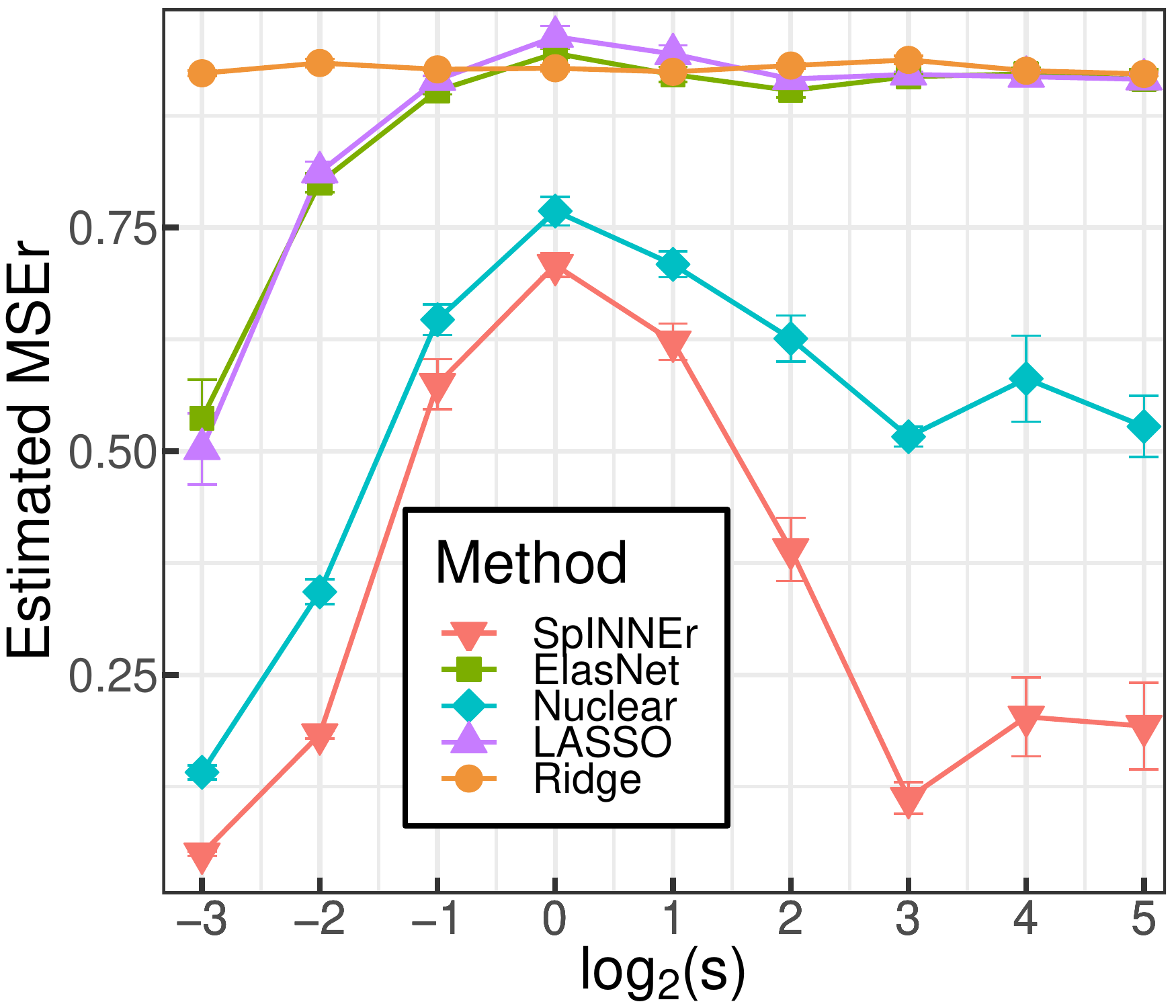}
  \caption{Scenario 1}
  \label{fig_Scenario1}
\end{subfigure}%
\
\begin{subfigure}{.325\textwidth}
  \centering
  \includegraphics[width=1\linewidth]{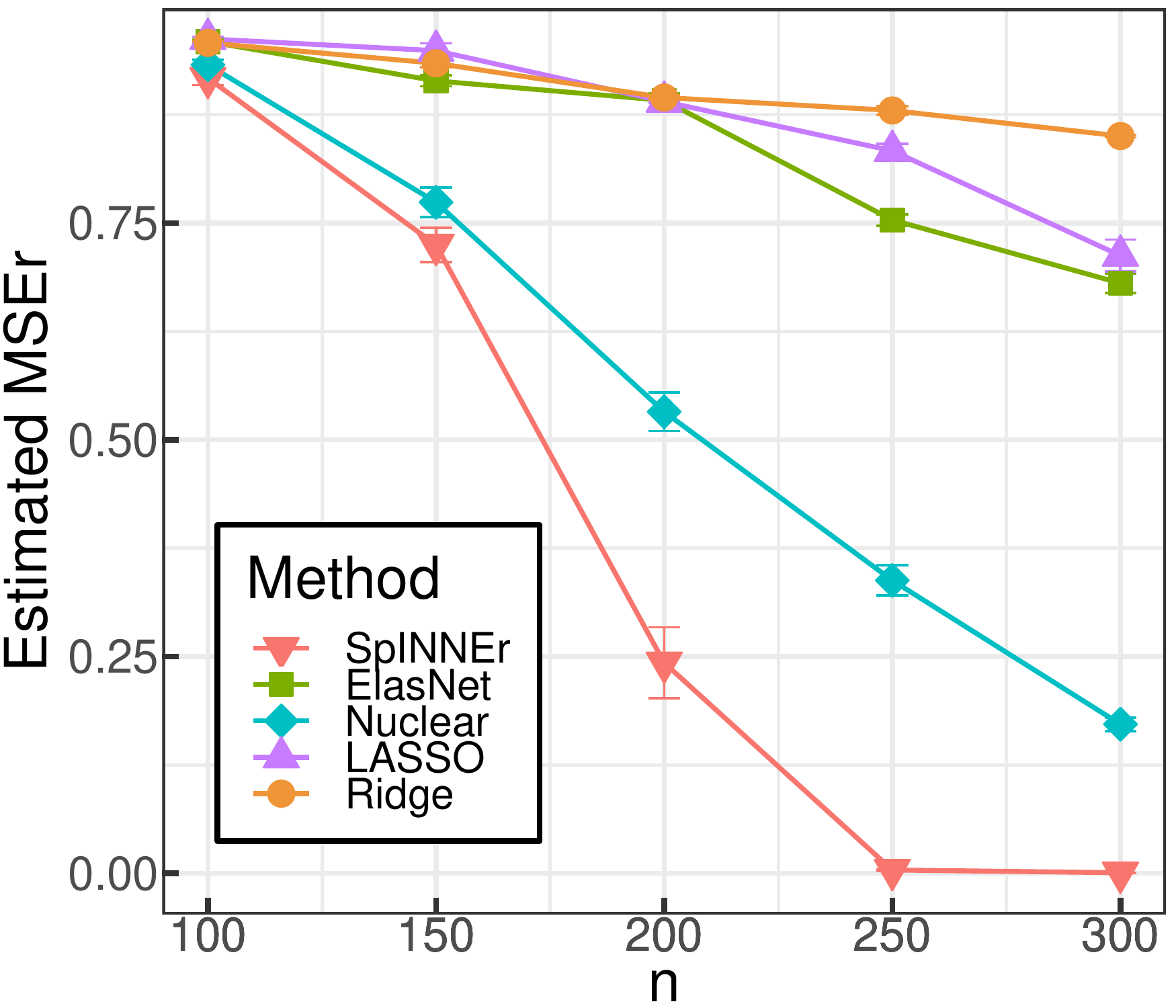}
  \caption{Scenario 2}
  \label{fig_Scenario2}
\end{subfigure}
\begin{subfigure}{.325\textwidth}
  \centering
  \includegraphics[width=1\linewidth]{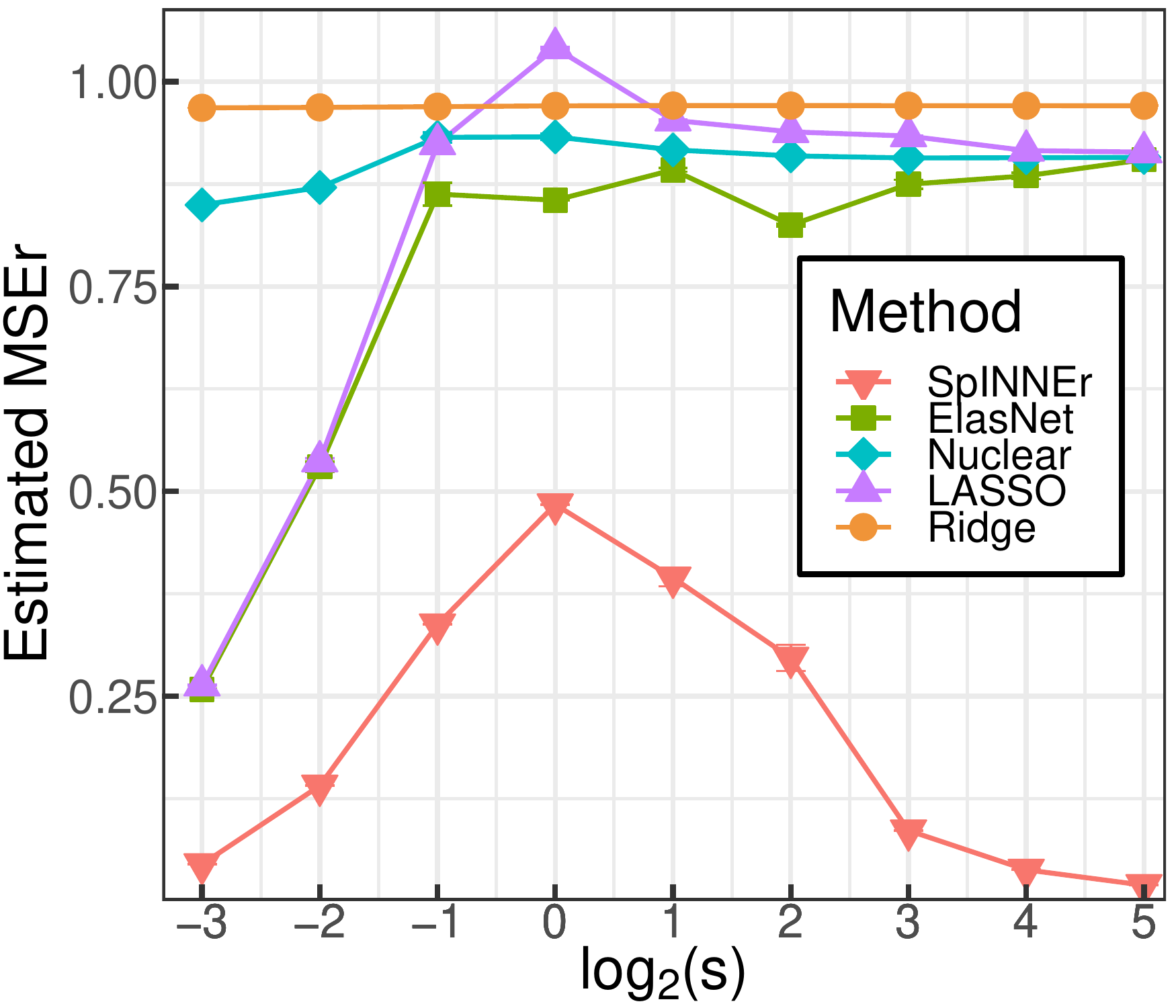}
  \caption{Scenario 3}
  \label{fig_Scenario3}
\end{subfigure}
\caption{Relative mean squared errors (MSEr) of estimators obtained from $\SP$, elastic net (ElasNet), nuclear-norm regression (Nuclear), lasso, and ridge under different simulation scenarios. Each point represents the average MSEr over 100 replicates and error bars indicates $95\%$ confidence intervals. (a) MSEr against $\log_2(s)$ under simulation \ref{S1}. (b) MSEr against sample size $n$ under simulation \ref{S2}. (c) MSEr against $\log_2(s)$ under simulation \ref{S3}.}
\label{fig_SimulationResults}
\end{figure}

As $s$ increases to values greater than 1, $\SP$ and nuclear-norm regression (the two methods that use a nuclear-norm penalty) exhibit substantial decrease in relative mean squared error. Elastic net and lasso, on the other hand, do not show a pronounced decrease. This demonstrates that when the true $B$ is both sparse and low-rank, and when the number of variables is comparable to the number of observations, encouraging sparsity alone is not sufficient for a regularized regression model to produce high estimation accuracy. Encouraging low-rank structure may be important. The behavior of ridge regression is different from the other four methods. As a shrinkage method that does not induce sparsity or low-rank structure, a ridge estimator's MSEr varies little with $s$.
\begin{figure}[ht]
\centering
\begin{subfigure}{.28\textwidth}
  \centering
	\includegraphics[width=1\linewidth]{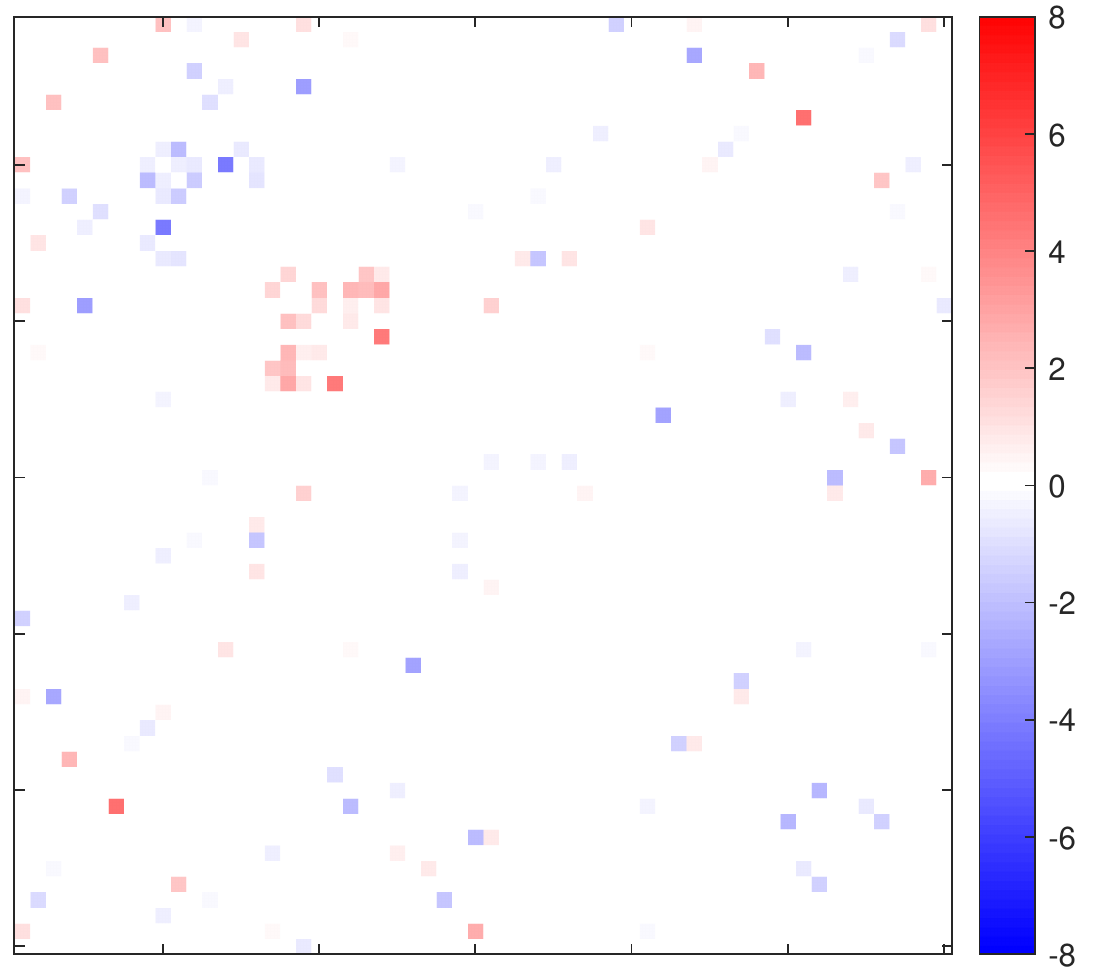}
  \caption{Elastic Net}
\end{subfigure}%
\ \ 
\begin{subfigure}{.28\textwidth}
  \centering
  \includegraphics[width=1\linewidth]{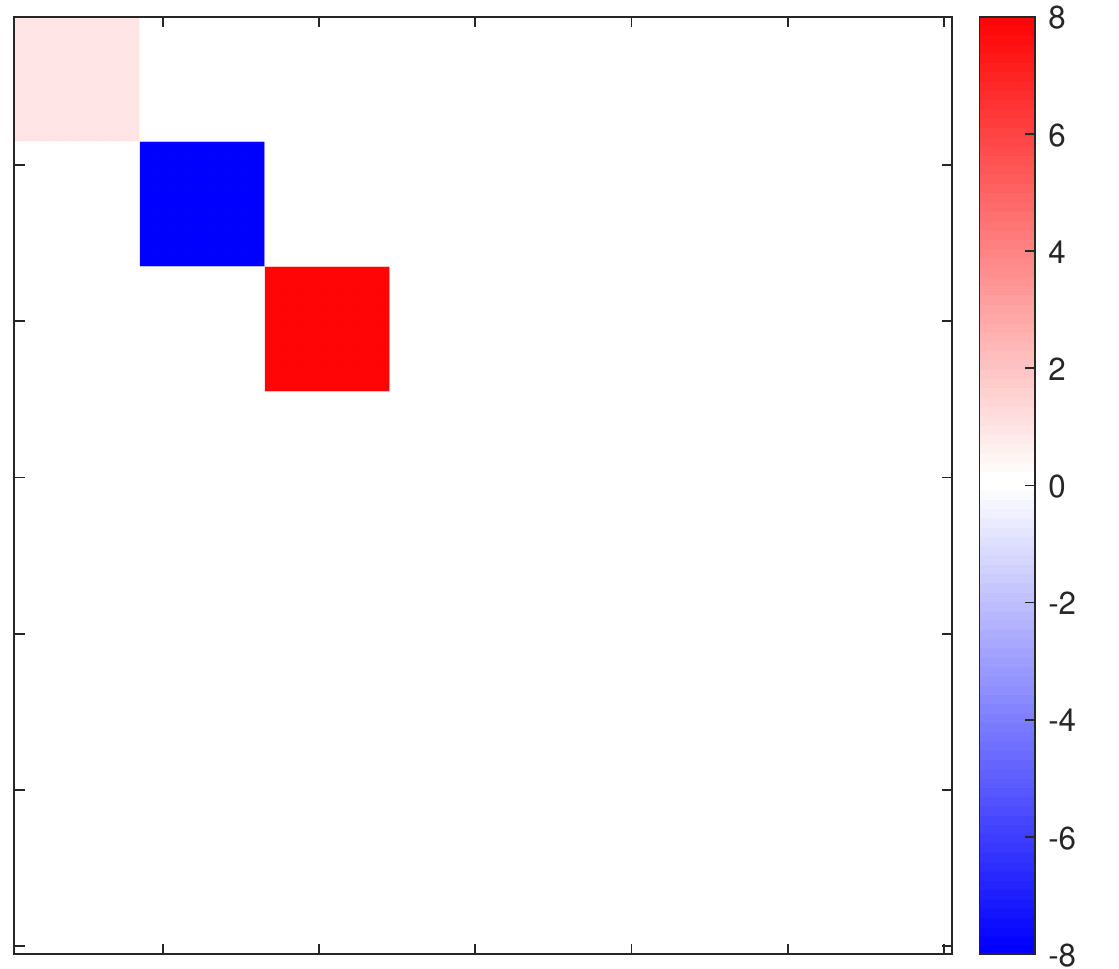}
  \caption{True $B$}
\end{subfigure}
\ \ 
\begin{subfigure}{.28\textwidth}
  \centering
  \includegraphics[width=1\linewidth]{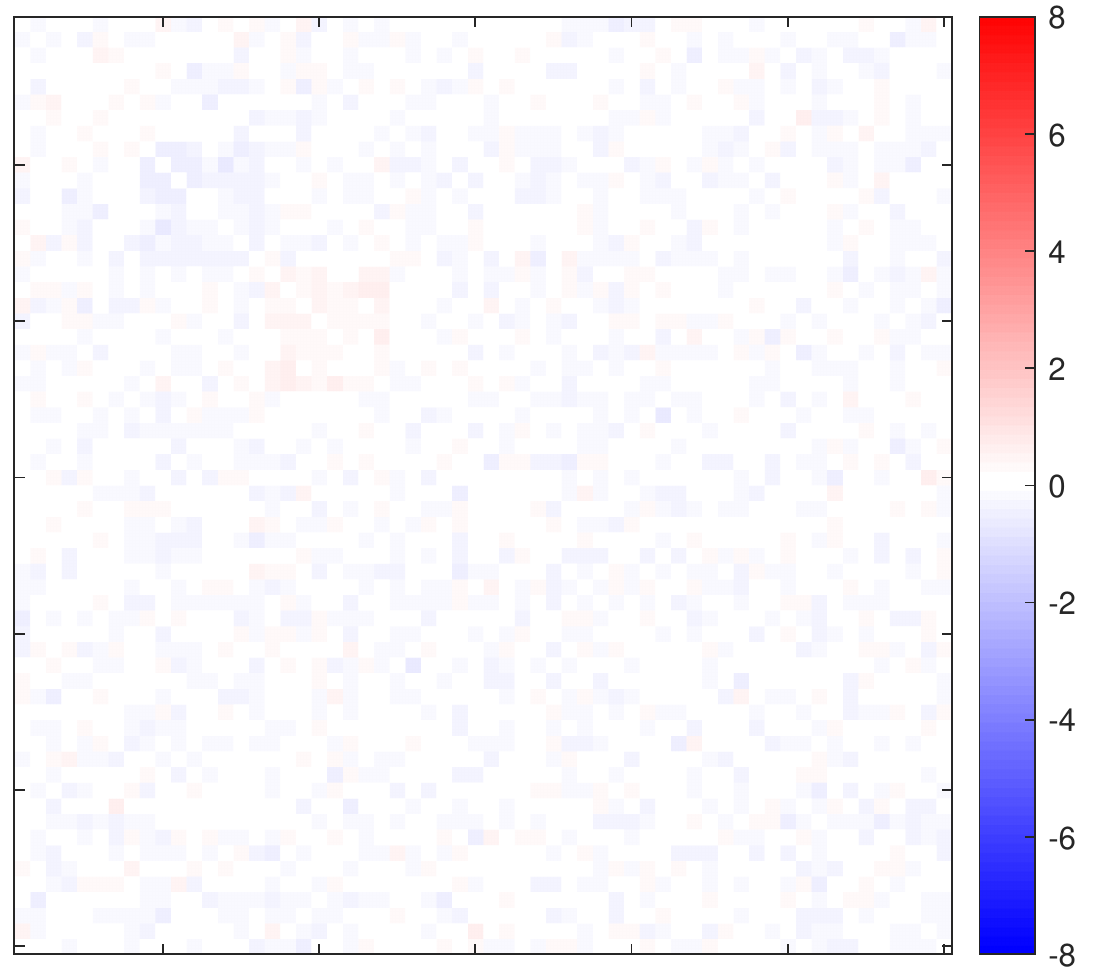}
  \caption{Ridge}
\end{subfigure}
\begin{subfigure}{.28\textwidth}
  \centering
	\includegraphics[width=1\linewidth]{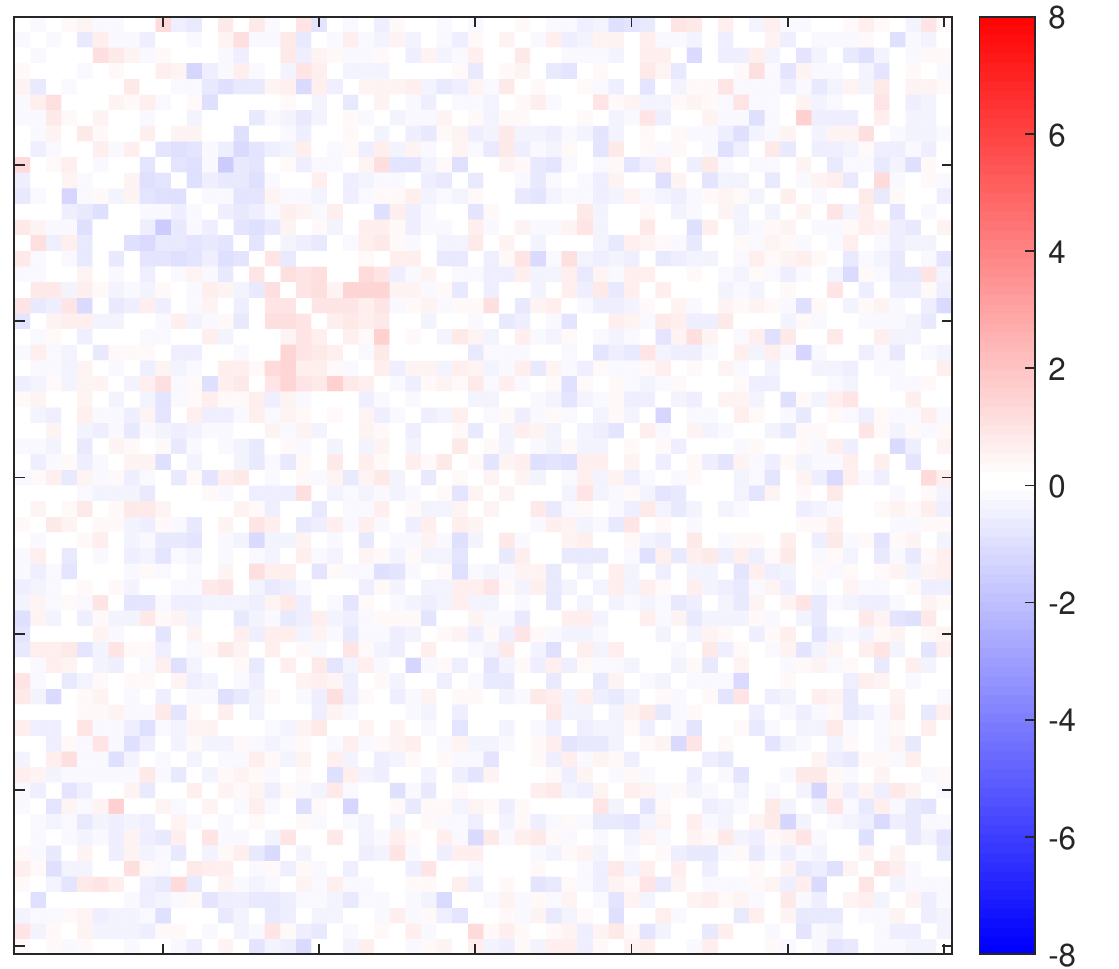}
  \caption{LASSO}
\end{subfigure}%
\ \ 
\begin{subfigure}{.28\textwidth}
  \centering
  \includegraphics[width=1\linewidth]{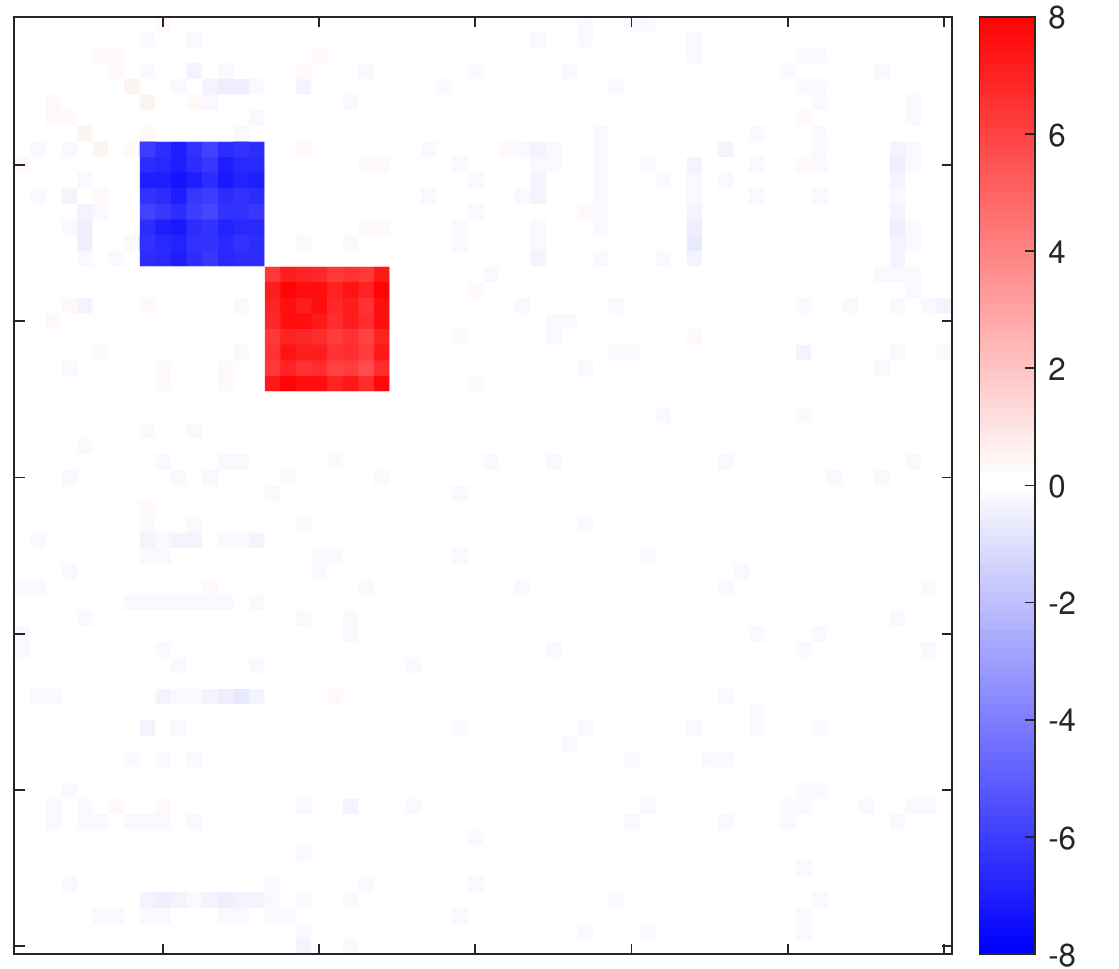}
  \caption{$\SP$}
\end{subfigure}
\ \ 
\begin{subfigure}{.28\textwidth}
  \centering
  \includegraphics[width=1\linewidth]{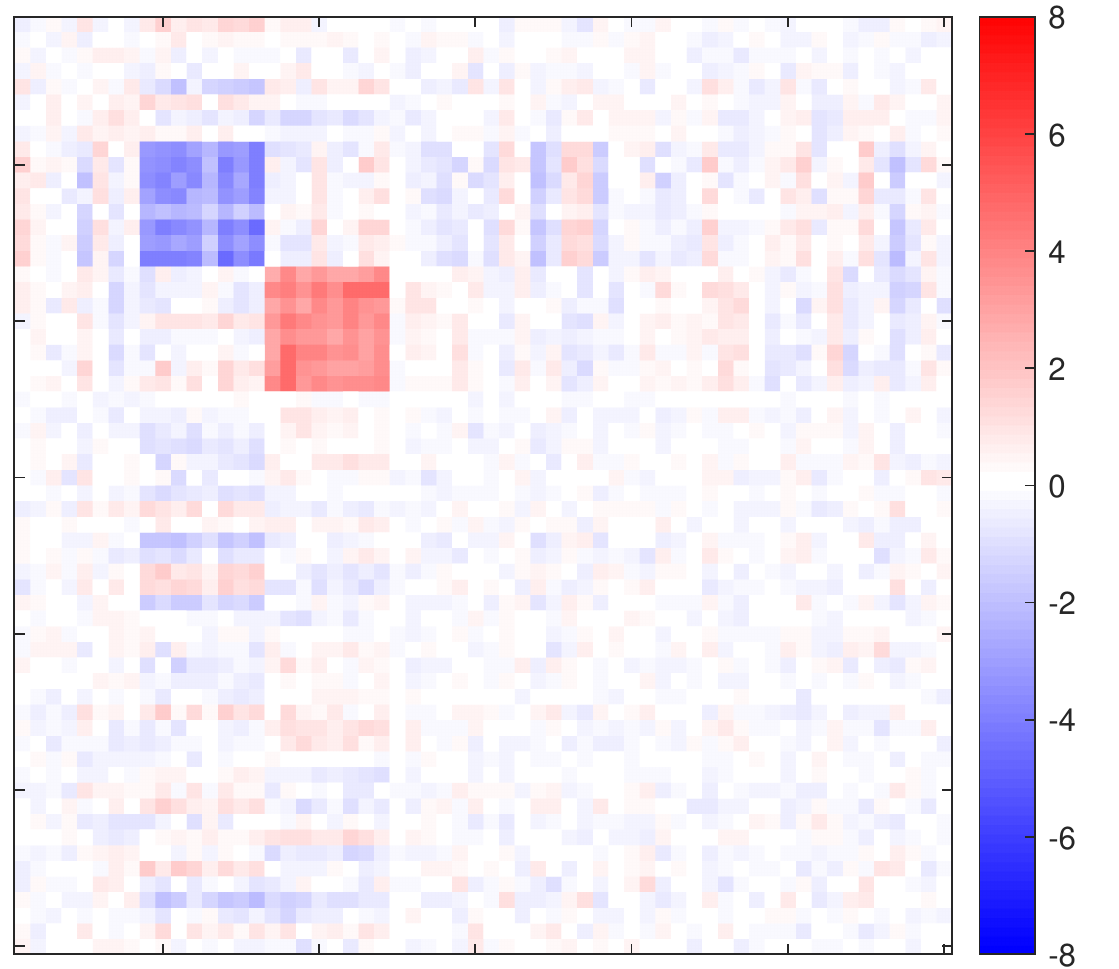}
  \caption{Nuclear}
\end{subfigure}
\caption{(b) True $B$ with $s = 8$ for simulation \ref{S1}. Each of (a), (c), (d), (e), and (f) shows the estimated $\hat{B}$ from one simulation run for each of the five regularization methods. The same color bar scale is shared across all subfigures.}
\label{fig_SimulationBhats}
\end{figure}

In a more focused look at these regularization methods, Figure~\ref{fig_SimulationBhats} displays the estimated $\hat{B}$ from a single simulation run---i.e., one representative set $\{y_i, A_i\}_{i=1}^n$---where the prescribed true $B$ consists of three signal-related blocks: $B_1 = \mathbbm{1}_{8\times8}$, \ $B_2 = -8\times\mathbbm{1}_{8\times8}$, and $B_3 = 8\times\mathbbm{1}_{8\times8}$, implying two positive and one negative RRC. Of these three blocks, $B_1$ is dominated by $B_2$ and $B_3$. We observe that while the elastic net estimate is sparse, and entries corresponding to these three blocks have the correct signs, the overall structure does not accurately recover the truth. For ridge regression, the estimate is neither sparse nor low-rank. For lasso, a relatively small tuning parameter was chosen by cross-validation and hence the estimate is not sparse, although the block structure of $\hat{B}_2$ and $\hat{B}_3$ is, to some extent, discernible. For nuclear-norm regression, while the blocks $\hat{B}_2$ and $\hat{B}_3$ are more pronounced, many entries outside these blocks are nonzero, especially along the rows and columns of $\hat{B}_2$ and $\hat{B}_3$. Conversely, $\SP$ recovers $B_2$ and $B_3$ effectively, and although a few nonzero entries outside the three blocks are estimated to be non-zero, their magnitudes are small, hence producing an estimate having the smallest MSEr among the five methods. This example demonstrates how the simultaneous combination of low rank and sparsity penalization can recover this structure accurately while applying each penalty separately fails to do so.

\subsubsection{Scenario 2: Synthetic Connectivity Matrices with Varying Sample Sizes}\label{ssec:scenario2}
The results of simulation \ref{S1} show that when the sample size is fixed at 150, the MSEr for $\SP$ and nuclear-norm regression decreases substantially for $s>1$, while MSEr for all other models changes minimally, even for $s=128$.

Figure~\ref{fig_Scenario2} display the MSEr for $\hat B$ from each of the five estimates for $n \in \{50, 100, \ldots, 300\}$ and $s=1$ (\ref{S2}). The nonzero entries in $B$ consist of blocks $\mathbbm{1}_{8\times8}$, $-\mathbbm{1}_{8\times8}$, and $\mathbbm{1}_{8\times8}$, resulting in 3 RRC (each having 8 nodes) and 192 individual response-relevant variables. We can see that for all five methods, MSEr decreases with sample size, suggesting that each of these methods benefits from more information. However, the MSEr with $\SP$ and nuclear-norm regression decreases much faster than for the other methods which do not involve a nuclear-norm penalty. More specifically, for elastic net and lasso, their  when there are 300 observations are about the same as those from $\SP$ and nuclear-norm regression when there are only 150 observations.

Further, as seen in \ref{S1}, it is the simultaneous combination of the nuclear norm and $\ell 1$ penalties that is most effective.
In particular, although the decrease of MSEr for nuclear-norm regression appears to be at a rate that does not change much with $n$, the decrease for $\SP$ from $n = 150$ to $n = 200$ is much more substantial than the decrease from $n = 100$ to $n = 150$, suggesting that once the sample size exceeds the number of variables (192 in this case) $\SP$ may exhibit a leap in estimation accuracy. Moreover, as the sample size increases beyond 250, the relative mean squared error from $\SP$ is nearly zero, while more than 300 observations for nuclear-norm regression to approach zero (when $n = 300$, MSEr is still approximately 0.2).

\subsubsection{Scenario 3: Real Connectivity Matrices with Varying Signal Strength}\label{ssec:scenario3}
Figure~\ref{fig_Scenario3} shows MSEr from the five regularization methods under \ref{S3}, where the matrices $\{A_i\}_{i=1}^n$ represent functional connectivity among brain regions estimated from $n=100$ humans. To simulate signal, we once again considered 3 RRC by assigning nonzero entries in diagonal blocks of $B$, $\mathbbm{1}_{6\times6}$, $-s\times\mathbbm{1}_{5\times5}$, and $s\times\mathbbm{1}_{8\times8}$.  In this scenario, $\SP$ has lower MSEr than all other methods across all values of $s$. As in \ref{S1}, the MSEr for $\SP$ and lasso are at their highest when $s = 1$, which gives $125$ response-relevant variables in a sample size of $n=100$.

When $s>1$, while the relative mean squared error for $\SP$ decreases with $s$, the error curves for the other methods are relatively flat. This is similar to the results in Figure~\ref{fig_Scenario1}, except for the nuclear-norm regularization. A closer examination of the $\hat{B}$ from nuclear-norm regularization under \ref{S1} (see Figure~\ref{fig_SimulationBhats}(f)) and \ref{S3} (not shown) reveals that although the solutions are not sparse, the estimated blocks for RRCs are more pronounced under \ref{S1} than under \ref{S3}, most likely because the response-relevant entries constitute a larger fraction of the true $B$ under \ref{S1}. Finally, we note that the error bars in Figure~\ref{fig_Scenario3} are narrower than in Figure~\ref{fig_Scenario1} because under \ref{S1}, different synthetic connectivity matrices are generated across replicates, while under \ref{S3}, where we use real functional connectivity matrices, only the response values in $y$ differ across replicates.

In summary, all simulation scenarios examined here demonstrate that $\SP$ significantly outperforms elastic net, nuclear-norm regression, lasso, and ridge in terms of MSEr.

\section{Application in Brain Imaging}
\label{sec_BrainDataResults}
Here we report on the results of $\SP$ as applied to a real brain imaging data set. The goal is to estimate the association of functional connectivity with neuropsychological (NP) language test scores in a cohort $n=116$ HIV-infected males. The clinical characteristics of this cohort are summarized in Table~\ref{rda_demographic_data_summary}.
\begin{table}[h]
  \centering
  \footnotesize
  \begin{tabular}{@{\extracolsep{1pt}}ccccccccccc}
  \\[-1.8ex]\hline
  \hline \\[-1.8ex]
  \textbf{Characteristic} && \multicolumn{1}{c}{\textbf{Min}} && \multicolumn{1}{c}{\textbf{Median}} && \multicolumn{1}{c}{\textbf{Max}} && \multicolumn{1}{c}{\textbf{Mean}} && \multicolumn{1}{c}{\textbf{StdDev}} \\
  \hline \\[-1.8ex]
	Age										&&	20	&&	51		&&	74				&&	46.5		&&	14.8	\\
	Recent VL							&&	20	&&	20		&&	288000	&&	9228		&&	38921	\\
	Nadir CD4							&&	0		&&	193		&&	690				&&	219.5		&&	171	\\
	Recent CD4						&&	20	&&	536		&&	1354			&&	559.1		&&	286.5	\\
  \\[-1.8ex] \hline
  \hline
  \end{tabular}
	\caption{Characteristics for 116 males included in the study. The term ``CD4'' refers to CD4 cells -- white blood cells fighting the virus. The number of these cells declines with the progress of HIV infection and the patient is diagnosed with AIDS when CD4 count drops below 200. The notation ``VL'' corresponds to the viral load -- the number of HIV particles in a milliliter of blood. HIV is labeled as undetectable for VR smaller than 200 copies/ml while a high VR is considered at the level of about 100\,000 copies/ml.}
	\label{rda_demographic_data_summary}
\end{table}

For each participant, their estimated resting state functional connectivity matrix, $A_i$, and age, $X_i$, are included in the regression model.  Each functional connectivity matrix, $A_i$ was constructed according to the Destrieux atlas (aparc.a2009s) \citep{Destrieux}, which defines $p=148$ cortical brain regions. The response variable, $y$, is defined as the mean of two word-fluency test scores: the  \textit{Controlled Oral Word Association Test-FAS} and the \textit{Animal Naming Test}.

We hypothesize that brain connectivity is associated with $y$ via a subset of the $148\times 148$ brain region connectivity values. As in \eqref{model} this is modeled as
\begin{equation}
\label{modelRDA}
y_i\  =\ \langle A_i, B \rangle \, +\, [1\, X_i]\left [\begin{BMAT}(c)[0.5pt,0pt,0.7cm]{c}{cc}\beta_1 \\ \beta_2 \end{BMAT} \right ]\, +\, \varepsilon_i, \quad i=1,\ldots, n,\qquad \textrm{for}\quad \varepsilon_i\sim \mathcal{N}\big(0, \sigma^2\big).
\end{equation}

The  $\SP$ estimate of $B$ comes from tuning parameters $\lambda_N = 12.1$ and $\lambda_L=2.4$. These were selected by the 5-fold cross-validation from $225$ grid points; i.e. all pairwise combinations of $15$ values of $\lambda_N$ and $15$ values $\lambda_L$ (see, subsection~\ref{subsec_SimulationImplementation}).
The connectivity matrices, the covariate of age and the response variable were all standardized across subjects before performing cross-validation. We attempted to fit both the lasso and nuclear-norm estimates by applying one-dimensional cross validation, but in each case the estimated $B$ matrices contained all zeros.  Therefore, for displaying these latter estimates, we used the marginal tuning parameter values from the two-dimensional $\SP$ cross validation.

Figure~\ref{Results} shows the three matrix regression estimates. The estimate from $\SP$ is in Figure~\ref{Results}(b) and is flanked by estimates from the lasso (with $\lambda_L=2.4$; Figure~\ref{Results}(a)) and the nuclear-norm penalty ($\lambda_N = 12.1$; Figure~\ref{Results}(c)). We also marked 7 main brain networks which were extracted and labeled in \cite{Yeo2011} (known as \textit{Yeo seven-network parcellation}).%
\begin{figure}[h]
\centering
\begin{subfigure}{.325\textwidth}
  \centering
	\includegraphics[width=1\linewidth]{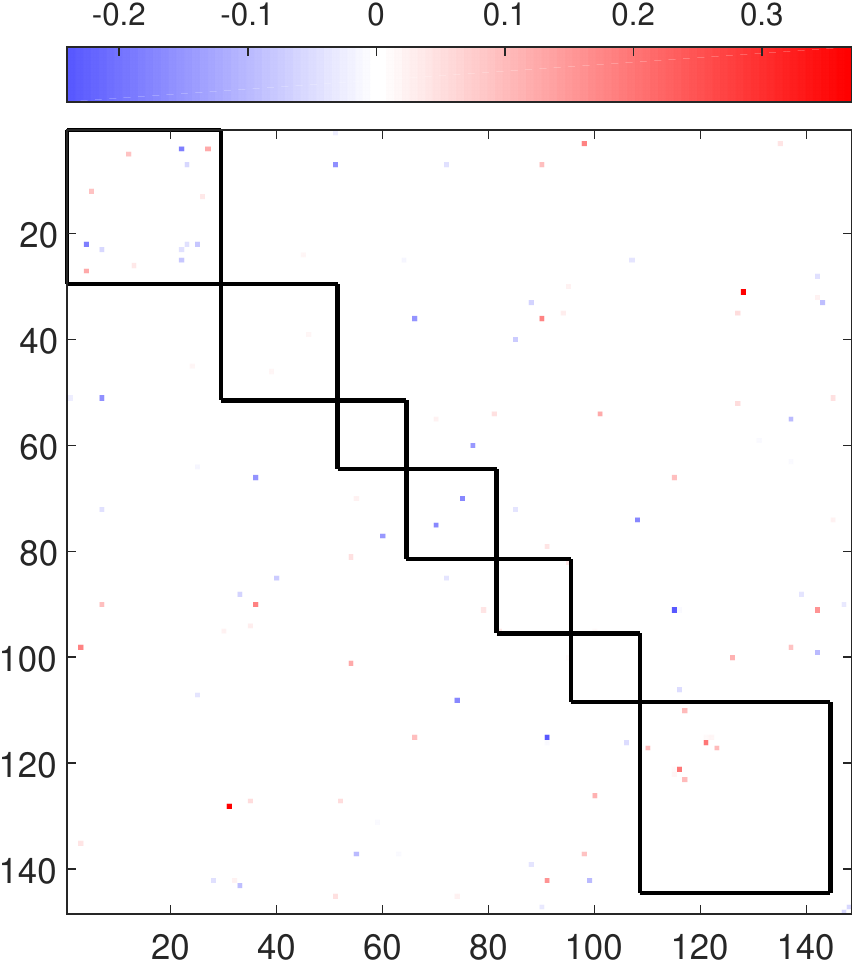}
  \caption{Lasso estimate}
\end{subfigure}%
\  
\begin{subfigure}{.325\textwidth}
  \centering
	\includegraphics[width=1\linewidth]{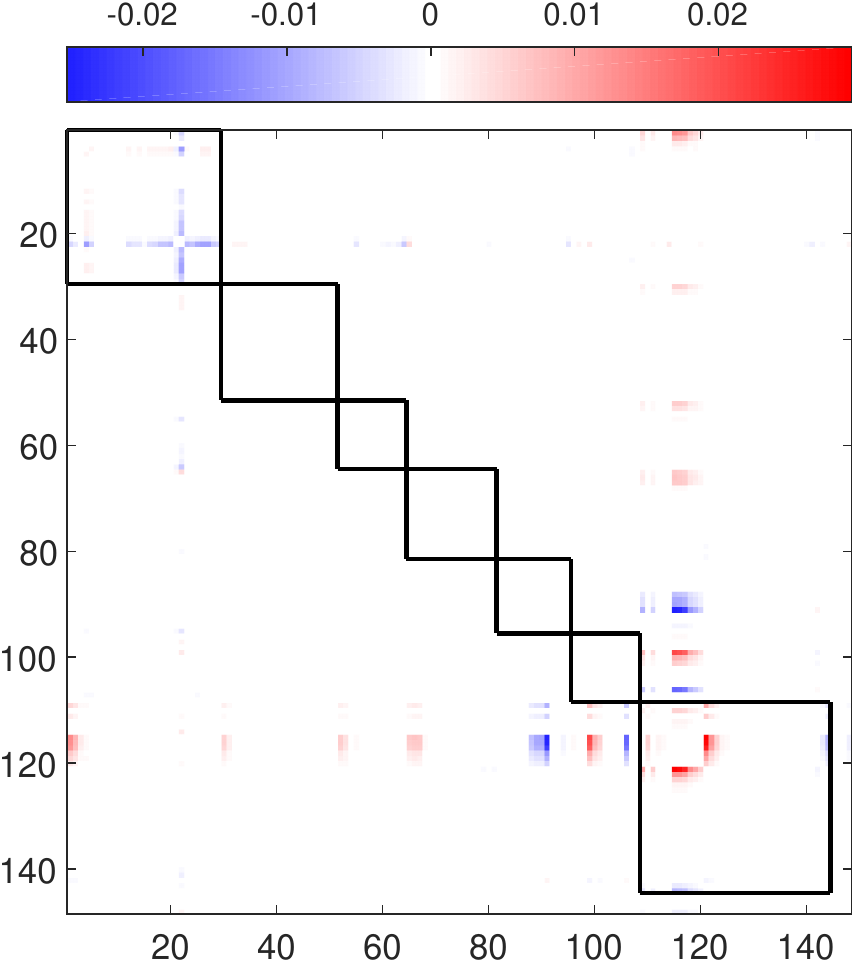}
  \caption{$\SP$ estimate}
\end{subfigure}
\ 
\begin{subfigure}{.325\textwidth}
  \centering
	\includegraphics[width=1\linewidth]{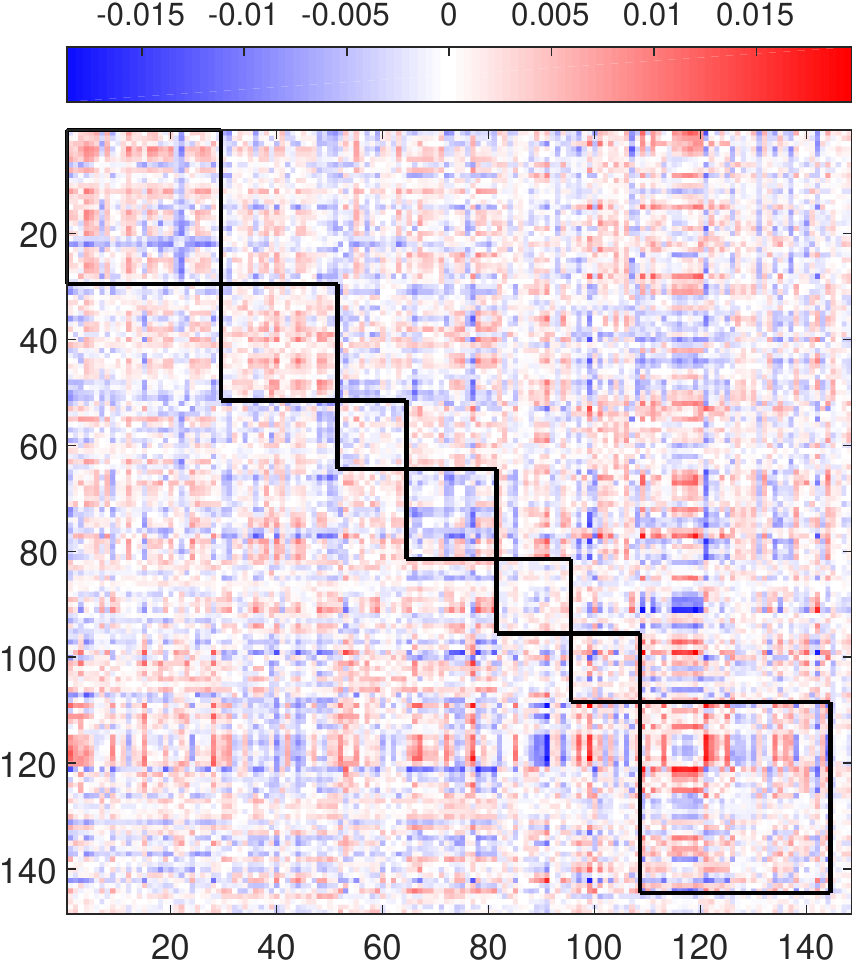}
  \caption{Nuclear-norm estimate}
\end{subfigure}
\caption{(b) presents $\SP$ estimate with $\lambda_N = 12.1$ and $\lambda_L = 2.4$ selected via 5-fold cross-validation, (a) shows the lasso solution with $\lambda_L = 2.4$ and (c) corresponds to nuclear-norm solution with $\lambda_N = 12.1$. The black boxes show the Yeo's parcellation into seven main brain networks. Nodes were permuted inside boxes in order to reveal the clusters.}
\label{Results}
\end{figure}
The graph in Figure~\ref{graphFigure}(b) reveals a very specific structure of the estimated associations. This structure is based on five brain regions which comprise the boundary between positive and negative groups of edges.  
\captionsetup[sub]{font=small,labelfont={bf,sf}}
\begin{figure}[h!]
\centering
\begin{subfigure}{.42\textwidth}
  \centering
  \includegraphics[width=1\linewidth]{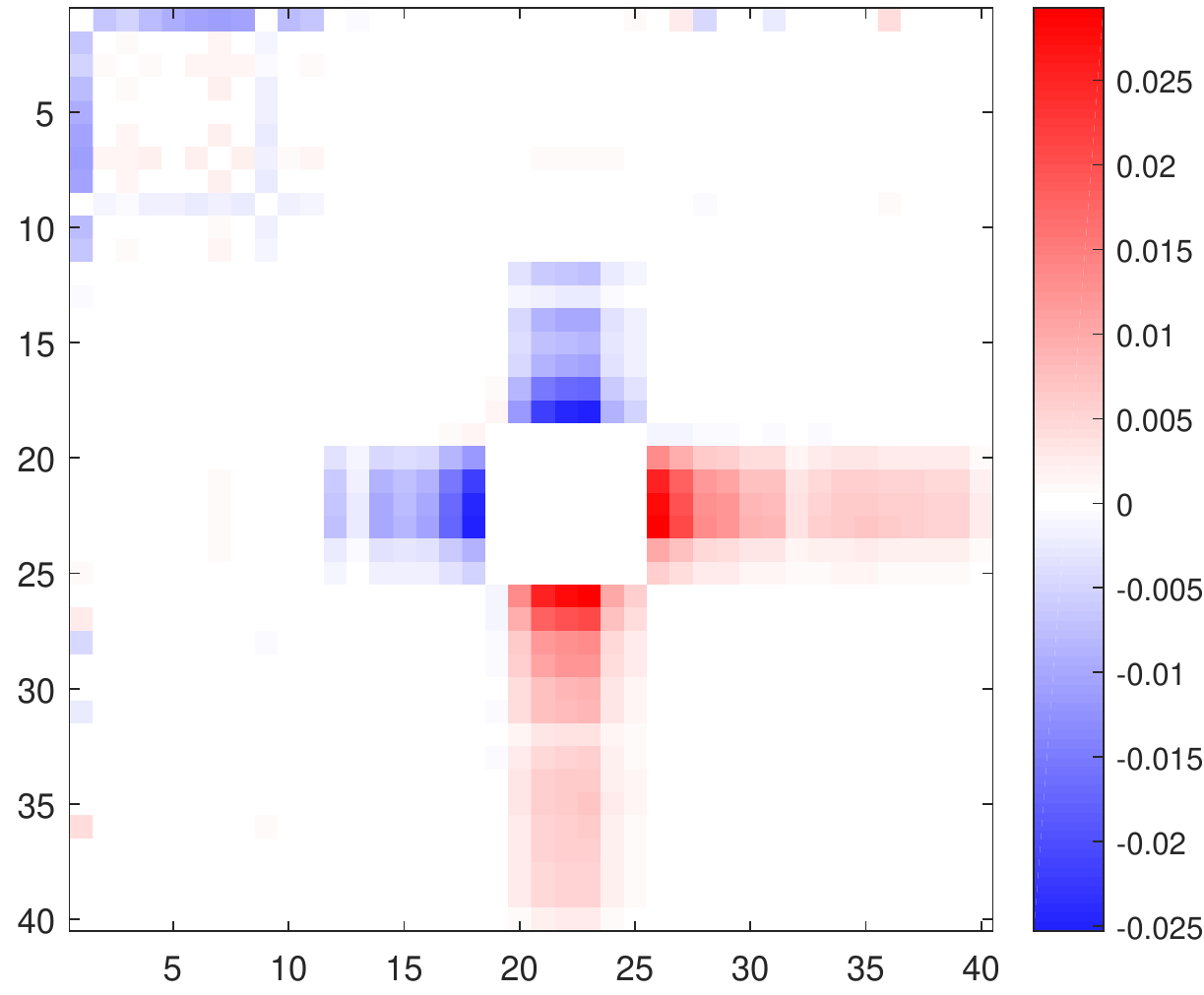}
  \caption{$\SP$ after nodes reordering}
\end{subfigure}
\ \
\begin{subfigure}{.45\textwidth}
  \centering
  \includegraphics[width=1\linewidth]{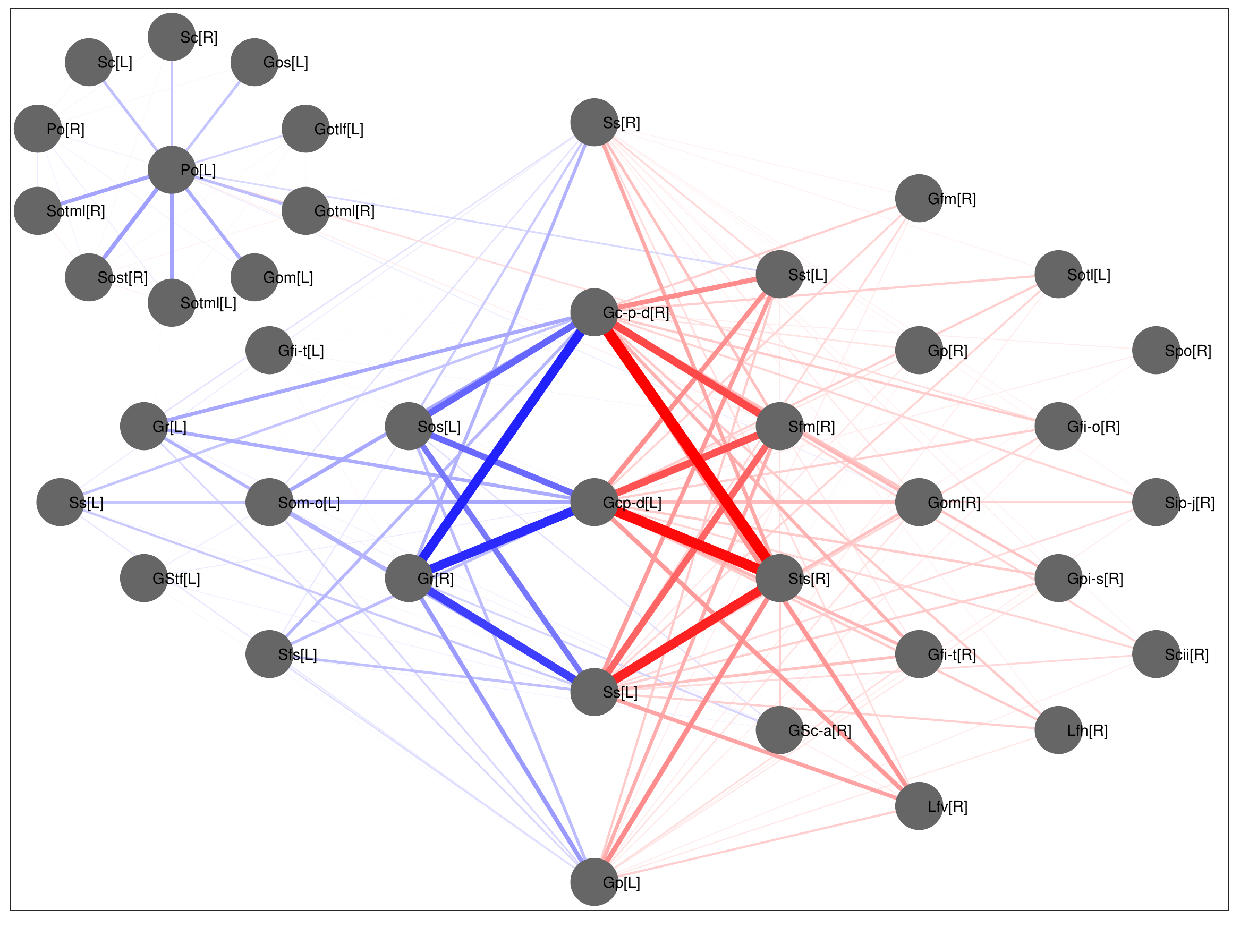}
  \caption{Graph representation}
\end{subfigure}
\caption{$\SP$ estimate restricted to 40 brain regions for which some response-related connectivities were found. (a) presents estimate after permuting nodes to achieve the cluster-by-cluster order. This was done based on the largest coefficients magnitudes of first a few left and right-singular vectors from singular value decomposition of $\hat{B}$, which indicate clusters indices. The corresponding graph representation is shown in (b).}
\label{graphFigure}
\end{figure}
These five regions are spread across the brain from the frontal lobe (left and right suborbital sulcus, Ss[L] and Ss[R]) to the area located by the corpus callosum (left and right posterior-dorsal part of the cingulate gyrus, Gcp-d[L], Gcp-d[R]), and up to the medial part of the parietal lobe (left precuneus gyrus, Gp[L]). They span two response-related groups of brain regions. First with conductivities having negative associations with the response is represented by left orbital H-shaped sulci (Sos[L]), left and right gyrus rectus (Gr[L], Gr[R]) and left medial orbital sulcus, Som-o[L]. The second, showing the positive associations, contains left superior occipital sulcus and transverse occipital sulcus, Sst[L], right middle frontal sulcus (Sfm[R]), superior temporal sulcus, Sts[R], right vertical ramus of the anterior segment of the lateral sulcus Lfv[R] and right middle occipital gyrus, Gom[R]. Interestingly, there is also the third response-relevant group of brain regions clearly visible in Figure~\ref{graphFigure}(b). It has a different structure than two aforementioned groups and forms star-shaped subgraph of negative effects, with the center in the occipital pole, Po[L]. Brain networks were visualized in Figure~\ref{brainView} by using BrainNet viewer (http://www.nitrc.org/projects/bnv/).
\captionsetup[sub]{font=small,labelfont={bf,sf}}
\begin{figure}[h!]
\centering
\begin{subfigure}{.3\textwidth}
  \centering
  \includegraphics[width=1\linewidth]{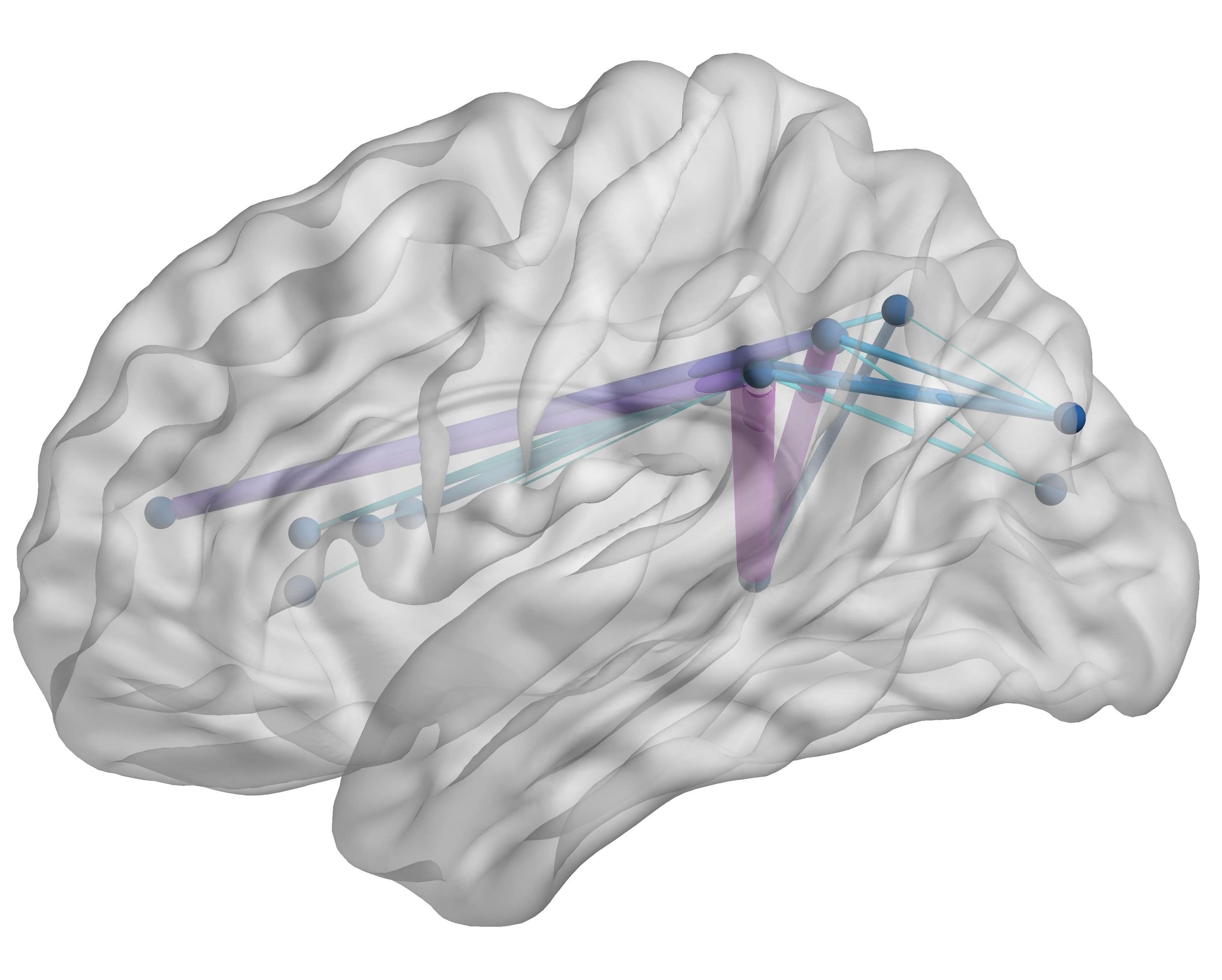}
  \caption{Sagittal view on PE}
\end{subfigure}
\
\begin{subfigure}{.2\textwidth}
  \centering
  \includegraphics[width=1\linewidth]{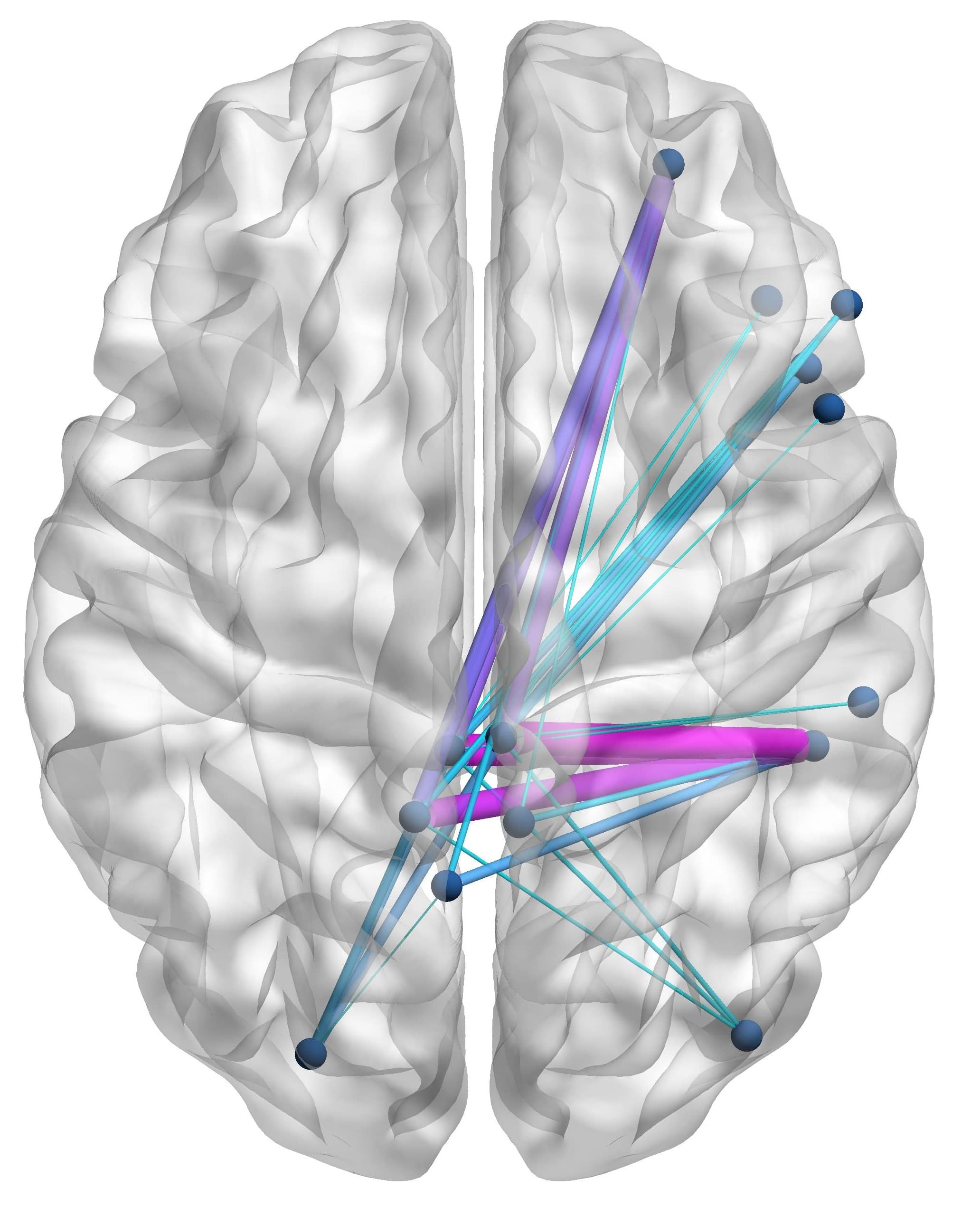}
  \caption{Axial view on PE}
\end{subfigure}
\
\begin{subfigure}{.31\textwidth}
  \centering
  \includegraphics[width=1\linewidth]{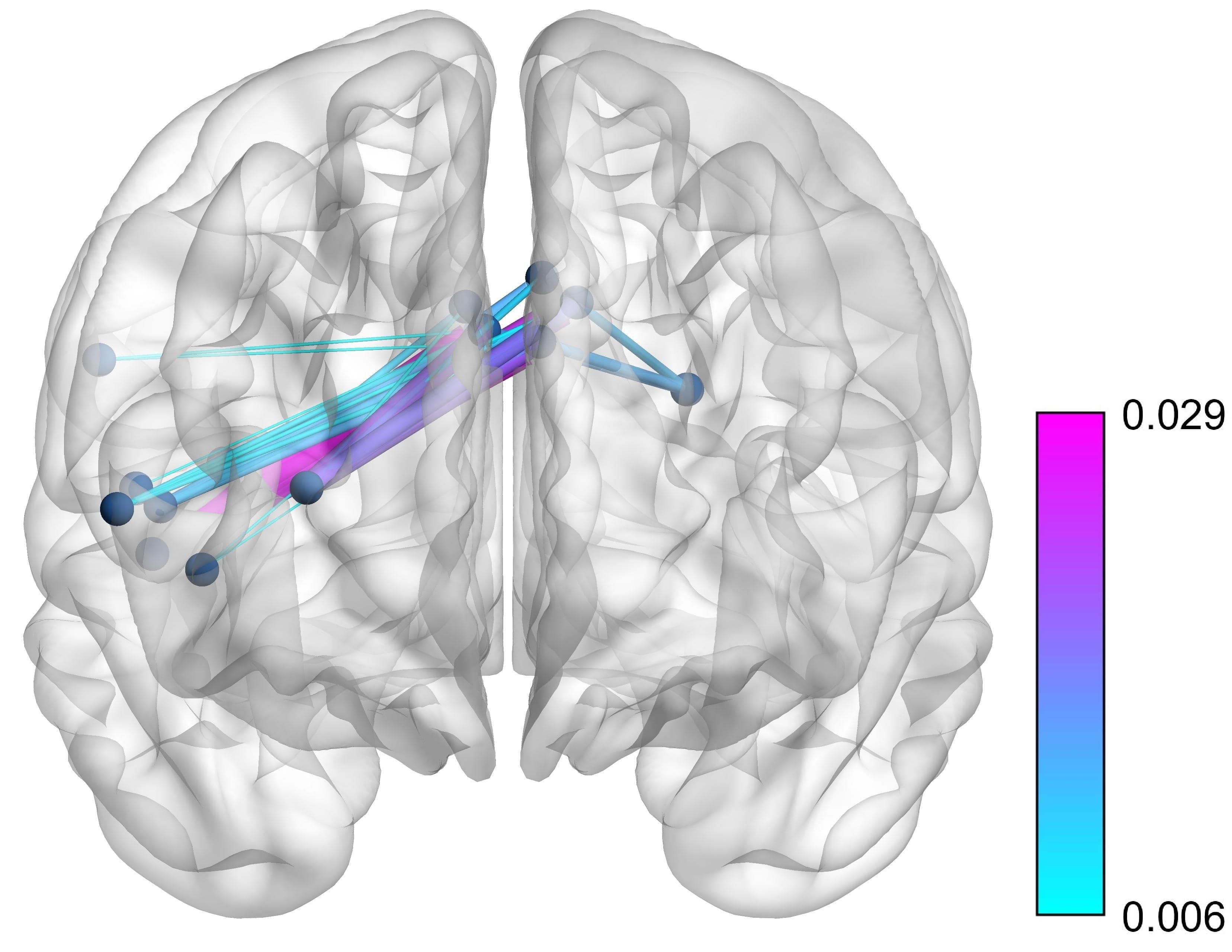}
  \caption{Coronal view on PE}
\end{subfigure}
\vspace{10pt}

\begin{subfigure}{.3\textwidth}
  \centering
  \includegraphics[width=1\linewidth]{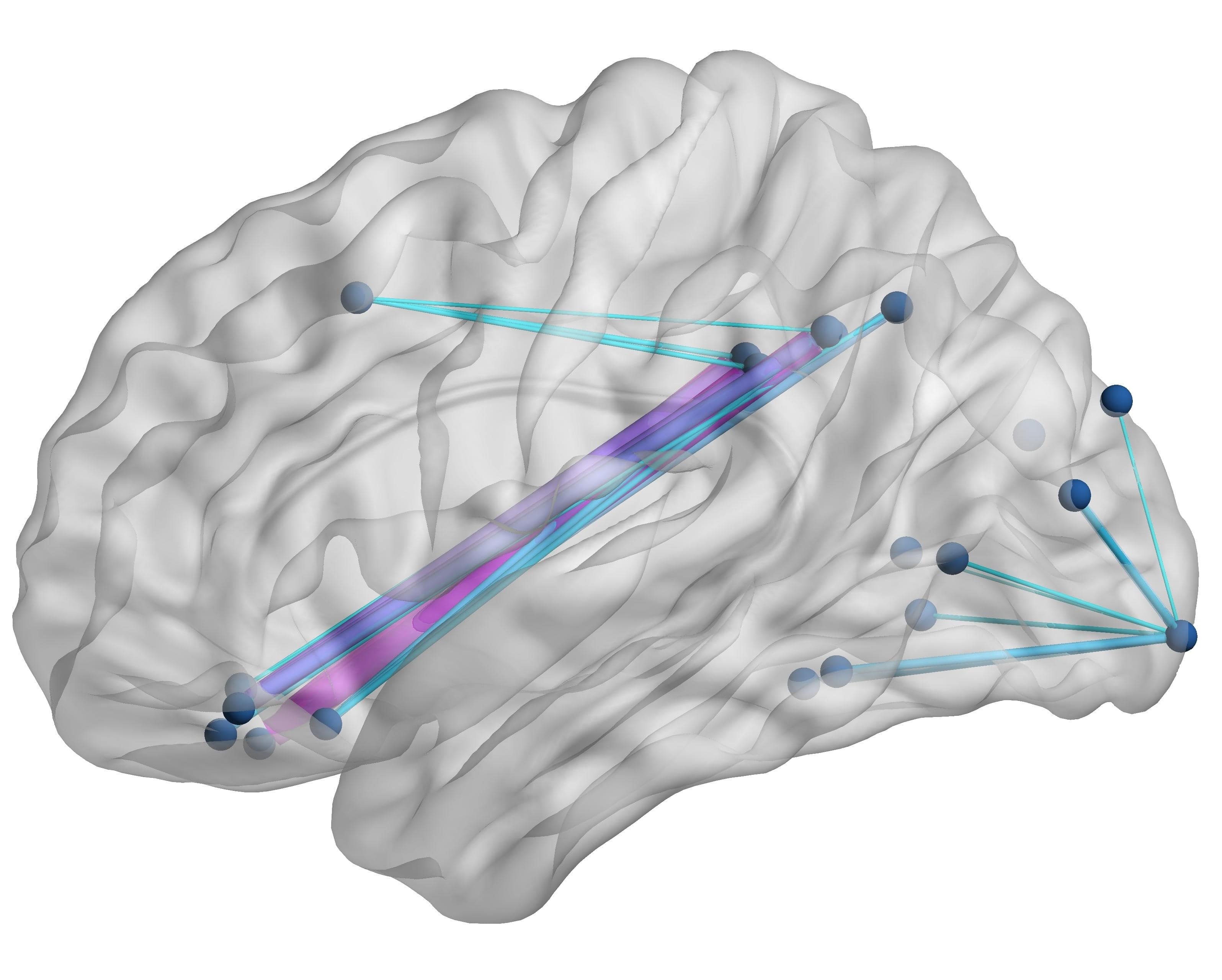}
  \caption{Sagittal view on NE}
\end{subfigure}
\
\begin{subfigure}{.2\textwidth}
  \centering
  \includegraphics[width=1\linewidth]{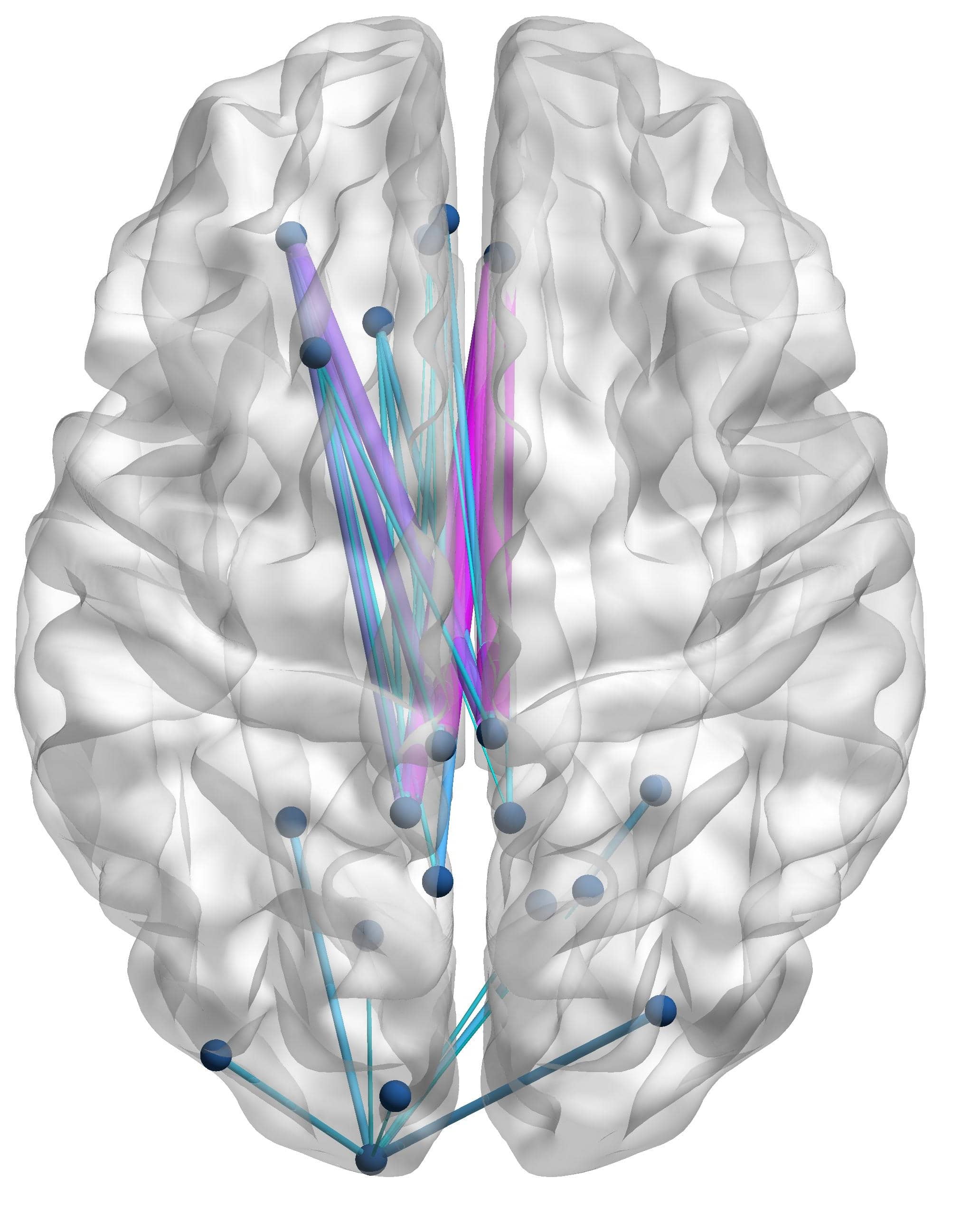}
  \caption{Axial view on NE}
\end{subfigure}
\
\begin{subfigure}{.31\textwidth}
  \centering
  \includegraphics[width=1\linewidth]{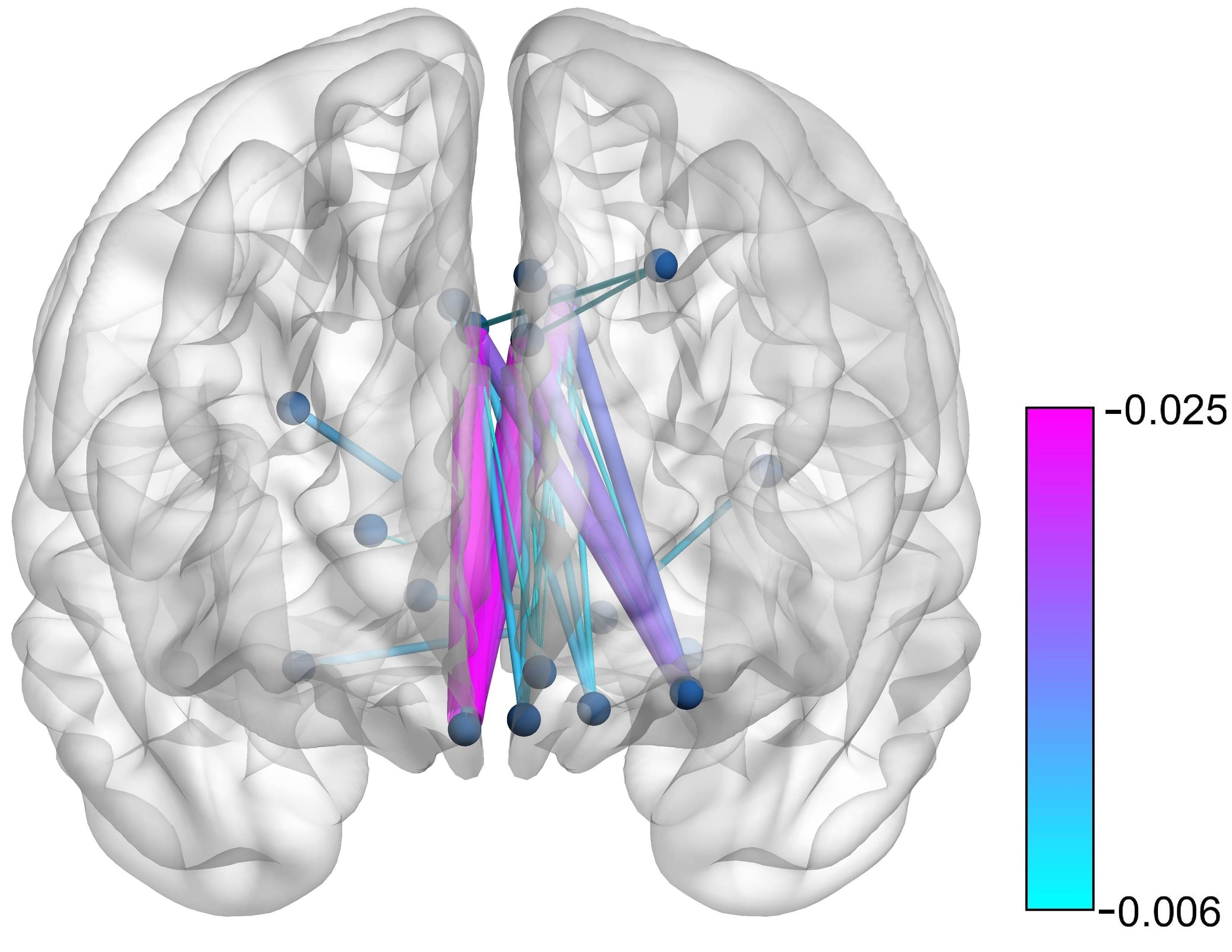}
  \caption{Coronal view on NE}
\end{subfigure}
\caption{The brain network visualization of response-relevant connectivities found by $\SP$. Different views on edges corresponding to positive entries of $\hat{B}^{\ES{S}}$ (PE) are presented in (a)--(c). Negative $\hat{B}^{\ES{S}}$ entries (NE) are shown in (d)--(f).}
\label{brainView}
\end{figure}

\section{Discussion}
\label{sec_Discussion}
We have proposed a novel way of estimating the regression coefficients for the scalar-on-matrix regression problem and derived theoretical properties of the estimator, which takes the form of a matrix. One of the primary contributions of this work is that it provides an accurate estimation of this matrix via a combination of two penalty terms: a nuclear norm and a $\ell_1$ norm. This approach may be viewed as an extension of both low-rank and sparse regression estimation approximations resulting in matrix estimates dominated by blocks structure. Advantages of our approach include: the estimation of meaningful, connected-graph regression coefficient structure; a computationally efficient algorithm via ADMM; and the ability to choose optimal tuning parameters.

In our simulation studies,  in Section~\ref{sec_SimulationExperiments} the first scenario illustrates the advantages of $\SP$ over several competing methods with respect to varying signal strengths. The second simulation scenario shows the performance of $\SP$ across a range of sample sizes. The third scenario shows how $\SP$ behaves when the data come from real structural connectivity matrices. In each case, $\SP$ outperforms all other methods considered.

Finally, we applied $\SP$ to an actual study of HIV-infected participants which aimed to understand the association of a language-domain outcome with functional connectivity. The estimated regression coefficients matrix revealed three response-related clusters of brain regions --- the smaller one dominated by a star-shaped structure of positive effects and two larger clumps (one negative and one positive) which shared 5 common brain regions. From the perspective of the RRCs recovery --- the notion which we introduce --- this can be treated as finding the overlapping clusters. However, $\SP$ can be also used to reveal much more complex structures than block diagonal matrices, which constitute a kind of the ``model signals'' for us. The class of sparse and low-rank signals includes also the matrices having symmetric, non-diagonal blocks, which may be important in some applications, like correlation matrices recovery.

Speed and stability are important considerations for the implementation of any complex estimation method. We have implemented the $\SP$ estimation process using the ADMM algorithm by dividing the original optimization problem (for given tuning parameters $\lambda_N$ and $\lambda_L$) into three subproblems, deriving their analytical solutions and computing them iteratively until the convergence. For each such iteration, our implementation precisely selects the step sizes based on the idea of residual balancing which turns out to work very fast and stable in practice. The final solution is obtained after cross-validation applied for the optimal selection of tuning parameters.

We note that the weights matrix $W$ must be prespecified at the beginning of the $\SP$ algorithm. The default setting (which we always used in this article) is a matrix of zeros on its diagonal and ones on off-diagonal entries. However, our implementation allows for an arbitrary choice of nonnegative weights. One may consider the selection based on the external information, if such is available, imposing weaker penalties for the entries being already reported as response-relevant in the particular application. The other possible strategy is an adaptive construction of $W$. It may rely on using the default $W$ first and update it, based on $\SP$ estimate, so as the large magnitude of $\hat{B}^{\ES{S}}_{j,l}$ generates small value of $W_{j,l}$. This procedure emphasizes the findings and can potentially improve variable selection accuracy.

In the future we want to perform a valid inference on the estimated clusters. Additionally, as developed here, $\SP$ addresses scalar-on-matrix regression models involving a continuous response. However, binary and count responses are often of interest. Indeed, an important problem that arises in studies of HIV-infected individuals is that of understanding the association of (binary) impairment status and neuro-connectivity. These more general settings will motivate future work in the estimation problem for scalar-on-matrix regression.

\section*{Declaration of interest}
The authors confirm that there are no known conflicts of interest associated with this publication and there has been no significant financial support for this work that could have influenced its outcome.
\section*{Acknowledgements}
DB, JH, TWR and JG were partially supported by the NIMH grant R01MH108467. D.B. was also funded by Wroclaw University of Science and Technology resources (8201003902, MPK: 9130730000). Data were provided [in part] by the Human Connectome Project, WU-Minn Consortium (Principal Investigators: David Van Essen and Kamil Ugurbil; 1U54MH091657) funded by the 16 NIH Institutes and Centers that support the NIH Blueprint for Neuroscience Research; and by the McDonnell Center for Systems Neuroscience at Washington University.
\vspace{20 pt}

\bibliographystyle{apalike}
\bibliography{references}

\begin{thebibliography}{}

\bibitem[Basu et~al., 2018]{Basu-low-rank-2018}
Basu, S., Li, X., and Michailidis, G. (2018).
\newblock Low rank and structured modeling of high-dimensional vector
  autoregressions.
\newblock arXiv:1812.03568.

\bibitem[Cand{\`e}s and Recht, 2009]{Candes2009}
Cand{\`e}s, E.~J. and Recht, B. (2009).
\newblock Exact matrix completion via convex optimization.
\newblock {\em Foundations of Computational Mathematics}, 9(6):717.

\bibitem[Chandrasekaran et~al., 2012]{Chandrasekaran-latent-2012}
Chandrasekaran, V., Parrilo, P.~A., and Willsky, A.~S. (2012).
\newblock Latent variable graphical model selection via convex optimization.
\newblock {\em The Annals of Statistics}, 40(4):1935--1967.

\bibitem[Ciccone et~al., 2019]{Ciccone-robust-2019}
Ciccone, V., Ferrante, A., and Zorzi, M. (2019).
\newblock Robust identification of ``sparse plus low-rank'' graphical models:
  An optimization approach.
\newblock arXiv:1901.10613.

\bibitem[Destrieux et~al., 2010]{Destrieux}
Destrieux, C., Fischl, B., Dale, A., and Halgren, E. (2010).
\newblock Automatic parcellation of human cortical gyri and sulci using
  standard anatomical nomenclature.
\newblock {\em NeuroImage}, 53(1):1--15.

\bibitem[Fischl, 2012]{FreeS}
Fischl, B. (2012).
\newblock Freesurfer.
\newblock {\em NeuroImage}, 62:774--781.

\bibitem[Foti et~al., 2016]{Foti-sparse-2016}
Foti, N., Nadkarni, R., Lee, A. K.~C., and Fox, E.~B. (2016).
\newblock Sparse plus low-rank graphical models of time series for functional
  connectivity in meg.
\newblock SIGKDD Workshop on Mining and Learning from Time Series.

\bibitem[Friedman et~al., 2010]{glmnet}
Friedman, J., Hastie, T., and Tibshirani, R. (2010).
\newblock Regularization paths for generalized linear models via coordinate
  descent.
\newblock {\em Journal of Statistical Software}, 33(1):1--22.

\bibitem[Frobenius, 1912]{FrobTheorem}
Frobenius, G. (1912).
\newblock Ueber matrizen aus nicht negativen elementen.
\newblock {\em S.-B. Preuss Acad. Wiss.}, pages 456--477.

\bibitem[Gabay and Mercier, 1976]{Gabay1976ADA}
Gabay, D. and Mercier, B. (1976).
\newblock A dual algorithm for the solution of nonlinear variational problems
  via finite element approximation.
\newblock volume~2, pages 17--40.

\bibitem[Goldsmith et~al., 2014]{Goldsmith-smooth-2014}
Goldsmith, J., Huang, L., and Crainiceanu, C.~M. (2014).
\newblock Smooth scalar-on-image regression via spatial {B}ayesian variable
  selection.
\newblock {\em Journal of Computational and Graphical Statistics},
  23(1):46--64.
\newblock PMID: 24729670.

\bibitem[Hastie et~al., 2015]{Hastie-matrix-2015}
Hastie, T., Mazumder, R., Lee, J.~D., and Zadeh, R. (2015).
\newblock Matrix completion and low-rank {SVD} via fast alternating least
  squares.
\newblock {\em Journal of Machine Learning Research}, 16(1):3367--3402.

\bibitem[Hoerl and Kennard, 1970]{ridge}
Hoerl, A.~E. and Kennard, R.~W. (1970).
\newblock Ridge regression: biased estimation for nonorthogonal problems.
\newblock {\em Technometrics}, 12(1):55--67.

\bibitem[Kolda and Bader, 2009]{Kolda-tensor-2009}
Kolda, T. and Bader, B. (2009).
\newblock Tensor decompositions and applications.
\newblock {\em SIAM Review}, 51(3):455--500.

\bibitem[Li et~al., 2010]{Li-dimension-2010}
Li, B., Kim, M.~K., and Altman, N. (2010).
\newblock On dimension folding of matrix- or array-valued statistical objects.
\newblock {\em The Annals of Statistics}, 38(2):1094--1121.

\bibitem[Mazumder et~al., 2010]{Mazumder-spectral-2010}
Mazumder, R., Hastie, T., and Tibshirani, R. (2010).
\newblock Spectral regularization algorithms for learning large incomplete
  matrices.
\newblock {\em Journal of Machine Learning Research}, 11:2287--2322.

\bibitem[Recht et~al., 2010]{Rech2010}
Recht, B., Fazel, M., and Parrilo, P.~A. (2010).
\newblock Guaranteed minimum-rank solutions of linear matrix equations via
  nuclear norm minimization.
\newblock {\em SIAM Rev.}, 52(3):471--501.

\bibitem[Reiss and Ogden, 2007]{Reiss-functional-2007}
Reiss, P.~T. and Ogden, R.~T. (2007).
\newblock Functional principal component regression and functional partial
  least squares.
\newblock {\em Journal of the American Statistical Association},
  102(479):984--996.

\bibitem[Reiss and Ogden, 2010]{Reiss-functional-2010}
Reiss, P.~T. and Ogden, R.~T. (2010).
\newblock Functional generalized linear models with images as predictors.
\newblock {\em Biometrics}, 66(1):61--69.

\bibitem[Rockafellar, 1970]{Rockafellar1970}
Rockafellar, R.~T. (1970).
\newblock {\em Convex analysis}.
\newblock Princeton, N.J. : Princeton University Press.

\bibitem[Tibshirani, 1996]{LassoF}
Tibshirani, R. (1996).
\newblock Regression shrinkage and selection via the lasso.
\newblock {\em Journal of the Royal Statistical Society: Series B},
  58(1):267--288.

\bibitem[Tibshirani, 2013]{tibshirani2012lasso}
Tibshirani, R. (2013).
\newblock The lasso problem and uniqueness.
\newblock {\em Electronic Journal of Statistics}, 7:1456--1490.

\bibitem[van~den Heuvel and Sporns, 2011]{Richclub}
van~den Heuvel, M.~P. and Sporns, O. (2011).
\newblock Rich-club organization of the human connectome.
\newblock {\em Journal of Neuroscience}, 31(44):15775--15786.

\bibitem[Van~Essen et~al., 2013]{conn}
Van~Essen, D.~C., Smith, S.~M., Barch, D.~M., Behrens, T. E.~J., Yacoub, E.,
  and Ugurbil, K. (2013).
\newblock The wu-minn human connectome project: An overview.
\newblock {\em NeuroImage}, 80:62--79.

\bibitem[Wang et~al., 2014]{Wang-regularized-2014}
Wang, X., Nan, B., Zhu, J., and Koeppe, R. (2014).
\newblock Regularized {3D} functional regression for brain image data via haar
  wavelets.
\newblock {\em The Annals of Applied Statistics}, 8(2):1045--1064.

\bibitem[Wohlberg, 2017]{Wohlberg2017ADMMPP}
Wohlberg, B. (2017).
\newblock Admm penalty parameter selection by residual balancing.
\newblock {\em ArXiv}, abs/1704.06209.

\bibitem[Xu et~al., 2017a]{Xu2017AdaptiveRA}
Xu, Z., Figueiredo, M. A.~T., Yuan, X., Studer, C., and Goldstein, T. (2017a).
\newblock Adaptive relaxed admm: Convergence theory and practical
  implementation.
\newblock {\em 2017 IEEE Conference on Computer Vision and Pattern Recognition
  (CVPR)}, pages 7234--7243.

\bibitem[Xu et~al., 2017b]{Xu:2017}
Xu, Z., Taylor, G., Li, H., Figueiredo, M. A.~T., Yuan, X., and Goldstein, T.
  (2017b).
\newblock Adaptive consensus admm for distributed optimization.
\newblock In {\em Proceedings of the 34th International Conference on Machine
  Learning - Volume 70}, pages 3841--3850.

\bibitem[Yeo et~al., 2011]{Yeo2011}
Yeo, B.~T., Krienen, F.~M., Sepulcre, J., Sabuncu, M.~R., Lashkari, D.,
  Hollinshead, M., Roffman, J.~L., Smoller, J.~W., Zollei, L., Polimeni, J.~R.,
  Fischl, B., Liu, H., and Buckner, R.~L. (2011).
\newblock The organization of the human cerebral cortex estimated by intrinsic
  functional connectivity.
\newblock {\em Neurophysiol}, 106(3):1125--65.

\bibitem[Zhao et~al., 2019]{HubDet}
Zhao, X., Wu, Q., Chen, Y., Song, X., Ni, H., and Ming, D. (2019).
\newblock Hub patterns-based detection of dynamic functional network metastates
  in resting state: A test-retest analysis.
\newblock {\em Frontiers in neuroscience}, 13(856).

\bibitem[Zhou and Li, 2014]{Zhou-regularized-2014}
Zhou, H. and Li, L. (2014).
\newblock Regularized matrix regression.
\newblock {\em Journal of the Royal Statistical Society: Series B (Statistical
  Methodology)}, 76(2):463--483.

\bibitem[Zhou et~al., 2013]{Zhou-tensor-2013}
Zhou, H., Li, L., and Zhu, H. (2013).
\newblock Tensor regression with applications in neuroimaging data analysis.
\newblock {\em Journal of the American Statistical Association},
  108(502):540--552.
\newblock PMID: 24791032.

\bibitem[Zou and Hastie, 2005]{elasticNet}
Zou, H. and Hastie, T. (2005).
\newblock Regularization and variable selection via the elastic net.
\newblock {\em Journal of the Royal Statistical Society: Series B},
  67(2):301--320.

\end{thebibliography}

\newpage
\appendix

\section*{Appendix}
\appendix
\section{Subproblems in ADMM algorithm}\label{subs:subproblems}
We use the letter $\mathcal{A}$ to denote the $n\times p^2$ matrix that collects the vectorized matrices, $A_1,\ldots, A_n$, in rows:  $\mathcal{A}:= \big [\vvec(A_1)|\ldots| \vvec(A_n)\big ]\T{T}$. The submatrix of $\mathcal{A}$ built from columns that correspond to the upper-diagonal entries of matrices $A_i$s (without diagonal entries) is denoted as $\mathcal{A}_U$. Accordingly, the columns of $\mathcal{A}$ that correspond to the symmetric entries from the lower-diagonal part is denoted by $\mathcal{A}_L$. Our implementation is derived under the assumption that all the matrices $A_i$s are symmetric and have zeros on their diagonals. Therefore, $\mathcal{A}_U = \mathcal{A}_L$.
\subsection{Analytical solution to \eqref{Update1}}
The considered update is
\begin{equation*}
\label{appendix1}
B^{[k+1]}:=\ \argmin {B}\bigg\{\,\sum_{i=1}^n\Big(y_i -\langle A_i, B \rangle\Big)^2 \ +\ \delta^{[k]}_1\Big\|\, B - D^{[k]}  - \frac{\dual_1^{[k]}}{\delta^{[k]}_1}\,\Big\|_F^2\,\bigg\}.
\end{equation*}
We introduce a variable $\widetilde{B}:= B - D^{[k]}  - \frac{\dual_1^{[k]}}{\delta^{[k]}_1}$, which gives $B = \widetilde{B} + D^{[k]}  + \frac{\dual_1^{[k]}}{\delta^{[k]}_1}$. This yields the equivalent problem
\begin{equation}
\label{eqRidge}
\widetilde{B}^*=\ \argmin {\widetilde{B}}\bigg\{\,\sum_{i=1}^n\Big(\underbrace{y_i -\langle A_i, D^{[k]}  + \frac{\dual_1^{[k]}}{\delta^{[k]}_1}\rangle}_{\widetilde{y}_i} - \langle A_i, \widetilde{B} \rangle\Big)^2 \ +\ \delta^{[k]}_1\big\|\,\widetilde{B}\,\big\|_F^2\,\bigg\}
\end{equation}
In terms of the stacked vectorized matrices in $\mathcal{A}$, \eqref{eqRidge} takes the form of ridge regression with the objective
\begin{equation*}
\begin{split}
\big\|\widetilde{y} -& \mathcal{A}\vvec(\widetilde{B})\big\|_2^2\ +\ \delta^{[k]}_1\big\|\vvec(\widetilde{B})\big\|_2^2 =\\
 &\big\|\widetilde{y} - \mathcal{A}_L\vvec_L(\widetilde{B}) - \mathcal{A}_U\vvec_U(\widetilde{B})\big\|_2^2\ +\ \delta^{[k]}_1\Big(\,\big\|\vvec_L(\widetilde{B})\big\|_2^2\ +\ \big\|\vvec_U(\widetilde{B})\big\|_2^2\ +\ \sum_{i=1}^p\widetilde{B}_{i,i}^{\,2}\,\Big),
\end{split}
\end{equation*}
where $\vvec_U(\widetilde{B})$ and $\vvec_L(\widetilde{B})$ are the vectors obtained from the upper and lower diagonal elements of $\widetilde{B}$, respectively.
\begin{proposition}
\label{prop2}
Suppose that all the matrices $A_i$'s have zeros on the diagonals and consider the minimization problem \eqref{eqRidge} with $\delta^{[k]}_1>0$. Then, $\widetilde{B}^*$ has zeros on the diagonal.
\end{proposition}
\begin{proof}
Suppose that $\widetilde{B}^*_{k,k}\neq 0$ for some $k\in\{1,\ldots,p\}$ and construct matrix $\overline{B}$ by setting $\overline{B}_{k,k}:=0$ and $\overline{B}_{i,j}:=\widetilde{B}^*_{i,j}$ for $(i,j)\neq (k,k)$. Obviously, we have that $\langle A_i, \widetilde{B}^* \rangle = \langle A_i, \overline{B} \rangle$ for each $i$. Denoting the objective function in \eqref{eqRidge} by $F$, we therefore get
\begin{equation*}
\begin{split}
F(\widetilde{B}^*) - F(&\overline{B})\ = \ \delta^{[k]}_1\big\|\,\widetilde{B}^*\,\big\|_F^2 - \delta^{[k]}_1\big\|\,\overline{B}\,\big\|_F^2 =\\
&\underbrace{\delta^{[k]}_1\sum_{(i,j)\neq(k,k)}(\widetilde{B}^*_{i,j})^2\, -\, \delta^{[k]}_1\sum_{(i,j)\neq(k,k)}\big(\overline{B}_{i,j}\big)^2}_{=\,0}\ +\ \delta^{[k]}_1\big(\widetilde{B}^*_{k,k}\big)^2\ >\ 0.
\end{split}
\end{equation*}
Consequently $F(\overline{B}) < F(\widetilde{B}^*)$, which contradicts the optimality of $\widetilde{B}^*$.
\end{proof}

Proposition~\ref{prop2} together with Proposition~\ref{prop1} imply that $\widetilde{B}^*$ is a symmetric matrix with zeros on the diagonal, which allows us to confine the minimization problem by the conditions $\vvec_L(\widetilde{B}) = \vvec_U(\widetilde{B}): = c$ and $\widetilde{B}_{i,i}=0$. This yields
\begin{equation}
\label{eqRidge2}
\vvec_U(\widetilde{B}^*)=\ \argmin {c\in\mathbb{R}^{(p^2-p)/2}}\bigg\{\,\big\|\widetilde{y} - 2\mathcal{A}_Uc\big\|_2^2\ +\ 2\delta^{[k]}_1\|c\|_2^2\bigg\}.
\end{equation}
In summary, it suffices to solve the ridge regression problem \eqref{eqRidge2} to obtain $\widetilde{B}^*$ and then recover $B^{[k+1]}$ by setting $B^{[k+1]} = \widetilde{B}^* + D^{[k]}  + \frac{\dual_1^{[k]}}{\delta^{[k]}_1}$. 

Now assume the (reduced) singular value decomposition (SVD) of $2\mathcal{A}_U$ is given; i.e., write $2\mathcal{A}_U = U\operatorname{diag}(d_1,\ldots, d_n)V\T{T}$, where $U\in \mathbb{R}^{n\times n}$ is an orthogonal matrix, $V\in \mathbb{R}^{(p^2-p)/2\times n}$ has orthogonal columns, and $d_1,\ldots, d_n$ are the $n$ singular values. The solution to \eqref{eqRidge2} can then be obtained as
\begin{equation*}
\label{ridgeSol}
\vvec_U(\widetilde{B}^*)\ =\ 2\Big(4\mathcal{A}_U\T{T}\mathcal{A}_U + 2\delta^{[k]}_1\mathbf{I}\Big)^{-1}\mathcal{A}_U\T{T}\widetilde{y}\ =\  V\left(\left [\begin{BMAT}(c)[0.5pt,0pt,0.7cm]{c}{ccc}d_1/\big(d_1^2+2\delta^{[k]}_1\big) \\\vdots \\d_n/\big(d_n^2+2\delta^{[k]}_1\big) \end{BMAT} \right ]\circ \big[U\T{T}\widetilde{y}\big]\right),
\end{equation*}
where ``$\circ$'' denotes a \textit{Hadamard product} (i.e., an entry-wise product of matrices). It is worth noting that the SVD of $\mathcal{A}_U$ need only be computed once, at the beginning of the numerical solver, since the left and right singular vectors, as well as the singular values, do not depend on the current iteration, nor do they depend on the regularization parameter in \eqref{eqRidge2} Therefore they can be used for the entire grid of regularization parameters in the $\SP$ process. This significantly speeds up the computation.

\subsection{Analytical solution to \eqref{Update2}}
We start with
\begin{equation*}
C^{[k+1]}:=\ \argmin {C}\bigg\{\,\frac12\Big\|\,\underbrace{D^{[k]}  + \frac{\dual_2^{[k]}}{\delta^{[k]}_2}}_{M^{[k]}}\,-\, C\,\Big\|_F^2\ +\ \frac{\lambda_N}{\delta^{[k]}_2}\|C\|_*\,\bigg\}.
\end{equation*}
To construct a fast algorithm for finding the solution, we use the following well known result.
\begin{proposition}
\label{propOrt}
For any matrix $M$ with singular value decomposition $M = U\operatorname{diag}(s)V\T{T}$, the optimal solution to
\begin{equation*}
\argmin{C} \Big\{\,\frac12\|C - M\|_F^2\ +\ \lambda \|C\|_*\,\Big\}
\end{equation*}
shares the same singular vectors as $M$ and its singular values are $s_i^*= (s_i - \lambda)_+:=\max\{s_i - \lambda, 0\}$.
\end{proposition}
Now, let $M^{[k]} = U^{[k]}\operatorname{diag}(s^{[k]}){V^{[k]}}\T{T}$ be the SVD of $M^{[k]}$. Thanks to Proposition \ref{propOrt}, we can recover $C^{[k+1]}$ in two steps
\begin{equation*}
\left\{
\begin{array}{l}
S^*:= \operatorname{diag}\Big(\,\big[\ (s^{[k]}_1 - \frac{\lambda_N}{\delta^{[k]}_2})_+\ ,\ldots,\ (s^{[k]}_p - \frac{\lambda_N}{\delta^{[k]}_2})_+\,\big]\T{T}\,\Big)\\
C^{[k+1]} = U^{[k]}S^*{V^{[k]}}\T{T}
\end{array}
\right..
\end{equation*}

\subsection{Analytical solution to \eqref{Update3}}
We use the following result.
\begin{proposition}
Let $D$, $K$ and $L$ be matrices with matching dimensions. Then, 
\begin{equation*}
\delta_1\big\|D -K\big\|_F^2\ +\ \delta_2\big\|D -L\big\|_F^2=(\delta_1+\delta_2)\bigg\|\,D -\frac{\delta_1K+\delta_2L}{\delta_1 + \delta_2}\,\bigg\|_F^2\ +\ \varphi(K, L, \delta_1, \delta_2),
\end{equation*}
where $\varphi(K, L, \delta_1, \delta_2)$ does not depend on $D$.
\end{proposition}
\begin{proof}
Simply observe that
\begin{equation*}
\begin{split}
&\nabla_D \bigg\{\delta_1\big\|D -K\big\|_F^2\ +\ \delta_2\big\|D -L\big\|_F^2 - (\delta_1+\delta_2)\Big\|D -\big(\delta_1K+\delta_2L\big)/\big(\delta_1+\delta_2\big)\Big\|_F^2\bigg\} =\\
&2\delta_1(D-K) + 2\delta_2(D-L) - 2(\delta_1+\delta_2)\Big[D - \big(\delta_1K+\delta_2L\big)/\big(\delta_1+\delta_2\big)\Big]\ =\\
&2\delta_1(D-K) + 2\delta_2(D-L) - 2(\delta_1+\delta_2)D + 2\big(\delta_1K+\delta_2L\big) =\ 0.
\end{split}
\end{equation*}
This proves the claim.
\end{proof}
Denote $\Delta^{[k]}: = \delta_1^{[k]} + \delta_2^{[k]}$. The above proposition reduces problem \eqref{Update3} to lasso regression under an orthogonal design matrix
\begin{equation*}
D^{[k+1]}=\ \argmin {D}\bigg\{\,\frac12\bigg\|\,\underbrace{\big(\delta_1^{[k]}B^{[k+1]}  + \delta_2^{[k]}C^{[k+1]}-\dual_1^{[k]} - \dual_2^{[k]}\big)/\Delta^{[k]}}_{Q^{[k+1]}}\,-\,D\,\bigg\|_F^2\ +\ \frac{\lambda_L}{\Delta^{[k]}}\Big\|\vvec(W \circ D)\Big\|_1\,\bigg\}.
\end{equation*}
The closed-form solution in this situation is well known and can be formulated simply as
\begin{equation*}
D^{[k+1]}_{ij}= \operatorname{sgn}\Big(Q^{[k+1]}_{ij}\Big)\cdot\bigg(\big|Q^{[k+1]}_{ij}\big|\, -\, \frac{\lambda_LW_{ij}}{\Delta^{[k]}}\bigg)_+,\quad\textrm{for}\quad i,j\in\{1,\ldots,p\}.
\end{equation*}
\section{Degenerate situations}\label{app:degenerate}
\vspace{-10pt}

\subsection{The case with $\lambda_L=0$}
We consider the problem \eqref{objective} with $\lambda_L=0$. We introduce a new variable, i.e., a $p\times p$ matrix $C$, to create the (equivalent) constrained  version of the problem with separable objective function:
\begin{equation*}
\argmin{B, C}\  \big\{f(B)\,+\,g(C)\big\} \qquad \textrm{s. t.}\ \ \ C-B=0.
\end{equation*}
The augmented Lagrangian with scalar $\delta>0$ and dual variable, $\dual\in \mathbb{R}^{p\times p}$, for this problem is
\begin{equation*}
L_{\delta}(B,C; \dual) =f(B)+g(C)+\langle\dual,\, C-B\rangle + \frac{\delta}2\big\|C-B\big\|_F^2,
\end{equation*}
and the ADMM updates for this case take the form
\begin{align}
&B^{[k+1]}:=\ \argmin {B}\bigg\{\,\sum_{i=1}^n\Big(y_i -\langle A_i, B \rangle\Big)^2\ +\ \delta^{[k]}_1\Big\|\,C^{[k]}  + \frac{\dual^{[k]}}{\delta^{[k]}_1}- B\,\Big\|_F^2\,\bigg\},\label{Update1nuc}\\
&C^{[k+1]}:=\ \argmin {C}\bigg\{\,\frac12\Big\|\,B^{[k+1]}- \frac{\dual^{[k]}}{\delta^{[k]}_1}-C\,\Big\|_F^2\ +\ \frac{\lambda_N}{\delta^{[k]}_1}\big\|C\big\|_*\,\bigg\},\label{Update2nuc}\\
&\dual^{[k+1]}: =\ \dual^{[k]} + \delta^{[k]}_1\big(C^{[k+1]}-B^{[k+1]}\big).
\end{align}
For $\lambda_L=0$ the criterion reduces to the method described by Zhou and Li in \cite{Zhou-regularized-2014}.
\vspace{-10pt}

\subsection{The case with $\lambda_N=0$}
We consider the problem \eqref{objective} with $\lambda_N=0$. We again introduce a new variable, i.e., a $p\times p$ matrix $D$, to create the (equivalent) constrained version of the problem with separable objective function:
\begin{equation*}
\argmin{B, D}\  \big\{f(B)\,+\,h(D)\big\} \qquad \textrm{s. t.}\ \ \ D-B=0.
\end{equation*}
The augmented Lagrangian with scalar $\delta>0$ and dual variable, $\dual\in \mathbb{R}^{p\times p}$, for this problem is
\begin{equation*}
L_{\delta}(B,C; \dual) =f(B)+h(D)+\langle\dual, D-B\rangle + \frac{\delta}2\big\|D-B\big\|_F^2,
\end{equation*}
and the ADMM updates for this case take the form
\begin{align}
&B^{[k+1]}:=\ \argmin {B}\bigg\{\,\sum_{i=1}^n\Big(y_i -\langle A_i, B \rangle\Big)^2\ +\ \delta^{[k]}_2\Big\|\,D^{[k]}  + \frac{\dual^{[k]}}{\delta^{[k]}_2}- B\,\Big\|_F^2\,\bigg\},\label{Update1nuc}\\
&D^{[k+1]}:=\ \argmin {D}\bigg\{\,\frac12\Big\|\,B^{[k+1]}- \frac{\dual^{[k]}}{\delta^{[k]}_2}-D\,\Big\|_F^2\ +\ \frac{\lambda_L}{\delta^{[k]}_2}\big\|\vvec(W\circ D)\big\|_1\,\bigg\},\label{Update2nuc}\\
&\dual^{[k+1]}: =\ \dual^{[k]} + \delta^{[k]}_2\big(D^{[k+1]}-B^{[k+1]}\big).
\end{align}
For $\lambda_N=0$ the criterion reduces to the lasso.

\end{document}